\documentclass[12pt,a4paper]{article}

\usepackage[T1]{fontenc}
\usepackage[utf8]{inputenc}
\usepackage{amsmath,amssymb,amsthm,mathtools}
\usepackage{bm}
\usepackage{braket}
\usepackage{graphicx}
\usepackage{tikz}
\usetikzlibrary{arrows.meta}
\usepackage{booktabs}
\usepackage[margin=1.1in]{geometry}
\usepackage[colorlinks,linkcolor=blue,citecolor=blue]{hyperref}
\usepackage{cleveref}

\theoremstyle{plain}
\newtheorem{theorem}{Theorem}[section]
\newtheorem{proposition}[theorem]{Proposition}
\newtheorem{lemma}[theorem]{Lemma}
\newtheorem{corollary}[theorem]{Corollary}
\theoremstyle{definition}
\newtheorem{definition}[theorem]{Definition}
\newtheorem{example}[theorem]{Example}
\theoremstyle{remark}
\newtheorem{remark}[theorem]{Remark}

\crefname{theorem}{theorem}{theorems}
\crefname{lemma}{lemma}{lemmas}
\Crefname{lemma}{Lemma}{Lemmas}
\Crefname{theorem}{Theorem}{Theorems}
\Crefname{proposition}{Proposition}{Propositions}
\Crefname{corollary}{Corollary}{Corollaries}
\Crefname{definition}{Definition}{Definitions}
\Crefname{remark}{Remark}{Remarks}
\crefname{proposition}{proposition}{propositions}
\crefname{corollary}{corollary}{corollaries}
\crefname{definition}{definition}{definitions}
\crefname{remark}{remark}{remarks}

\newcommand{\D}{\mathbb{D}}
\newcommand{\C}{\mathbb{C}}
\newcommand{\Z}{\mathbb{Z}}

\newcommand{\N}{\mathbb{N}}
\newcommand{\Ak}{A_k}
\newcommand{\Fock}{\mathcal{F}}
\newcommand{\Tk}{T_k}
\newcommand{\wk}{w_k}
\newcommand{\Ft}{\widetilde{F}}
\newcommand{\Finf}{F_\infty}
\newcommand{\ip}[2]{\langle #1 , #2 \rangle}
\newcommand{\avg}[1]{\langle #1 \rangle}
\newcommand{\norm}[1]{\lVert #1 \rVert}
\newcommand{\abs}[1]{\lvert #1 \rvert}

\newcommand{\ord}{\operatorname{ord}}
\newcommand{\dist}{\operatorname{dist}}
\newcommand{\Ree}{\operatorname{Re}}
\newcommand{\Imm}{\operatorname{Im}}
\newcommand{\Nz}{\mathcal{N}}   
\DeclareMathOperator{\Var}{Var}
\DeclareMathOperator{\Betadist}{Beta}
\newcommand{\rme}{\mathrm{e}}
\newcommand{\rmi}{\mathrm{i}}
\newcommand{\rmd}{\mathrm{d}}
\newcommand{\dA}{\,\rmd A}

\newcommand{\dth}{\,\rmd\theta}

\begin{document}
\sloppy   

\title{Stellar rank under the contraction of $SU(1,1)$\\
       to the Heisenberg--Weyl group}

\author{Chon-Fai Kam$^{1,2}$}
\date{}
\maketitle

\begin{center}
$^{1}$\,Dipartimento di Fisica e Chimica ``Emilio Segr\`e'', Universit\`a degli
Studi di Palermo,\\ Via Archirafi 36, I-90123 Palermo, Italy\\[2pt]
$^{2}$\,DSIMB, Inserm UMR\_S 1134 BIGR, Universit\'e Paris Cit\'e,\\
B\^atiment Lavoisier, 149 rue de S\`evres, 75015 Paris, France\\[6pt]
E-mail: \texttt{dubussygauss@gmail.com}
\end{center}
\medskip

\begin{abstract}
Under the contraction of $SU(1,1)$ to the Heisenberg--Weyl group, the zeros of
the Husimi function of a state redistribute according to a single scale, and
we determine which of them survive. The stellar rank of a single-mode bosonic
pure state, the number of these zeros, characterises the Gaussian states on
the plane through the vanishing of the rank. On the Poincar\'e disk, where the
same construction applies to the discrete series of $SU(1,1)$, it does not,
because the zero set leaves a zero-free factor undetermined. We therefore ask
how the zero structure behaves in the contraction limit rather than how it
should be classified at fixed Bargmann index. The entire difference between
the two Hilbert spaces reduces to one weight sequence
$w_k(m)=(2k)^m/(2k)_m$, whose monotonicity in $k$ drives every estimate below.
The rescaled states form a normal family, every nonzero limit point lies in
the Segal--Bargmann space, and those with finitely many zeros are exactly the
finite-rank Fock states. The rank is upper semicontinuous, with equality
precisely when all zeros scale as $O(1/\sqrt{2k})$, the scale of the
coherent-state labels. A two-sided Harnack estimate pins the profile on
compact sets. Separately, a family of rank $r$ degenerates if and only if the amplitudes on the lowest
$r+1$ weight vectors tend to zero. Mass escapes at the edge of the rescaled
disk instead, by a first-order transition of the semiclassical Husimi
potential at an explicit critical coupling.
Finally, the contraction limit retains exactly the lowest $2r+3$ amplitudes,
and this number cannot be lowered.
\end{abstract}


\section{Introduction}
\label{sec:intro}

A pure state of a single bosonic mode is represented in the Segal--Bargmann
space by an entire function, and the zeros of that function, equivalently the
zeros of the Husimi $Q$ function, form a discrete invariant. Chabaud, Markham
and Grosshans call their number the stellar rank and show that it is a
resource monotone for Gaussian operations, that it is unchanged by them, and
that it equals the minimal number of photon additions needed to prepare the
state \cite{Chabaud2020}. The bottom of the resulting hierarchy is the
statement that a state has rank zero if and only if it is Gaussian.

The mechanism behind it is a piece of classical function theory. A zero-free
entire function is
$\rme^{g}$, and the reproducing kernel of the Segal--Bargmann space bounds it
by $\rme^{\abs{z}^2/2}$. The Borel--Carath\'eodory inequality turns that bound
on $\Ree\,g$ into a bound on $\abs{g}$ and forces $\deg g\le2$, so the
resulting family $\rme^{az^2+bz+c}$ with $\abs{a}<1/2$ is exactly the set of
Gaussian states, which is Hudson's characterisation of the states with
non-negative Wigner function \cite{Hudson1974}. The rank has since been
developed as an operational resource, with certification protocols
\cite{Chabaud2021} and bounds on state conversion \cite{Hahn2026}; the sharp
form of the rank-$r$ wavefunction, $P(x)\rme^{-ax^2+bx+c}$, is established in
\cite{ComplexZeros2025}. Vourdas states this in physical form, observing that the
factors $\rme^{Q_0}$, $\rme^{Q_1}$ and $\rme^{Q_2}$ permitted by Hadamard's
theorem are the vacuum, a coherent state and a squeezed state, and noting
already there that the zeros do not determine the state, because Hadamard's
factorisation admits a zero-free exponential factor
\cite[section 5.3]{Vourdas2006}.

The same construction applies wherever a coherent-state family provides a
holomorphic representation, and in particular to the discrete series of
$SU(1,1)$, whose phase space is the Poincar\'e disk. Klimov, Mu\~noz, Leuchs,
Gazeau and S\'anchez-Soto construct the covariant family of $s$-ordered
quasiprobabilities there, prove that the Wigner function of every Perelomov
coherent state is strictly positive, and define the stellar rank through the
zeros of the Husimi function on the disk \cite{Klimov2026}. They show that
zeros imply Wigner negativity, and they also observe that the
converse fails, the
Barut--Girardello coherent state being zero-free on the disk while its Wigner
function is not positive, and they conclude that the zero set does not
determine the nonclassical structure of a state, because it does not see the
zero-free factor.

The phenomenon was noticed earlier, for $k=1/2$, where the representation
space is the Hardy space of the disk and Vourdas, Brif and Mann record it in
the language of inner and outer functions: the analytic representation of a
Barut--Girardello state is $\rme^{u^*z}$, which is an outer function and
therefore has no zeros in the disk, while the inner and outer parts of a
general state are separated by a factorisation in which only the inner part
carries zeros \cite{VBM1996}. The observation that the zeros leave the outer
and singular factors undetermined appears again in
\cite[section 10.4]{Vourdas2006}. The recent work extends the phenomenon to
all $k$ and supplies the quasiprobability framework in which the failure
acquires an operational meaning.

We do not attempt to repair the classification at fixed $k$. The question here
is what the zeros do under the group contraction $SU(1,1)\to H_1$, which
deforms the discrete series into the Schr\"odinger representation as
$k\to\infty$. The disk becomes the plane, the Bergman weight becomes the
Gaussian weight, and a criterion that is incomplete at every finite $k$ becomes
complete in the limit. Our results describe how, and in particular which zeros
survive the passage and which are lost.

The limit $k\to\infty$ is of course a familiar one. In the Berezin
formalism on the hyperbolic plane, $1/k$ plays the role of Planck's constant
and commutators of operator symbols tend to Poisson brackets
\cite[section 9.3]{Vourdas2006}. That limit acts on symbols of operators,
while the limit taken here acts on states and on the zeros of their Husimi
functions,
and neither statement implies the other, as \cref{sec:setting} makes
precise.

More substantially, Klimov and coauthors close \cite{Klimov2026} by arguing
that hyperbolic phase space supplies a self-consistent language for
quantumness which, in their words, ``does not need to borrow criteria designed
for flat or spherical geometries''. We agree, and the results below are an
argument for that position rather than against it.

The completeness of the flat criterion rests on the exponent two in the
Gaussian weight, which is what the Borel--Carath\'eodory step needs, and that
exponent is manufactured by the contraction and is absent at every finite $k$,
where the disk is a bounded domain and the zero-free class is
infinite-dimensional. The flat criterion is therefore
a feature of a limit rather than a standard that hyperbolic phase space fails
to meet, and \cref{sec:discussion} returns to this.

Group contraction in an analytic representation goes back at least to Ricci,
who contracts
the polynomial representations of $SU(2)$ to the Segal--Bargmann space,
rescaling the variable by $\varepsilon^{-1/2}$ and obtaining convergence of
the representations from the boundedness of the Fock norms together with the
Cauchy integral formula \cite{Ricci1986}. That argument is the same one that
gives normality in \cref{lem:normal}. Ricci does not consider zeros, and the
Majorana constellation had not then entered use as a description of a
representation, so the question could not be posed in his setting. The
reference is cited in the stellar-rank literature \cite{ChabaudMehraban2022}.

The zeros themselves reached phase space through quantum chaos, where Leboeuf
and Voros introduce the term stellar representation for the zeros of the Husimi
function and show that a pure state on a toric or spherical phase space is
determined by finitely many of them \cite{LeboeufVoros1990}, the companion
paper extending the construction to phase spaces generated by other Lie
groups \cite{LeboeufSaraceno1990}. Their statement
is for a fixed group, and no limit between groups is taken.

The two lines have met once, at the level of ensembles. Hannay computes the exact statistics of the Majorana zeros of a random spin-$j$
state and shows that as $j\to\infty$ they converge to the zeros of the planar
chaotic analytic function \cite{Hannay1996,Hannay1998}, the statistics of the
underlying random polynomials having been studied by Bogomolny, Bohigas and
Leboeuf \cite{BBL1996}. Bleher, Shiffman and Zelditch
prove the corresponding universality for random holomorphic sections on a
compact complex manifold, with the Bargmann--Fock kernel as the fixed point of
the rescaling \cite{BSZ2000}. Both statements concern correlation functions of
a Gaussian ensemble, and a state drawn from such an ensemble has infinite
stellar rank almost surely. The finite-rank states, which are the subject
here, form a set of measure zero that these results do not see.

The ensemble picture is nonetheless sharp at one point of the family, and the
contrast is worth drawing. The random model corresponding to $\Ak$ is the
hyperbolic Gaussian analytic function with intensity $L=2k$, and at $L=1$,
that is $k=1/2$, Peres and Vir\'ag prove that its zero set is a determinantal
point process, and the only member of the family for which this holds
\cite{PeresVirag2005}; see \cite{HKPV2009} for the general theory. The moduli of the zeros are then
independent, and the probability that a random state has
no zero in a disk of radius $r$ is the Euler product
$\prod_{m\ge1}(1-r^{2m})$. For other values of $k$ the corresponding hole
probabilities are known only through bounds
\cite{SodinTsirelson2005,BNPS2018}. Zero-freeness is thus completely
solved in the ensemble sense at that one value of $k$. What we describe is the
set of zero-free states, on which that probability measure puts no mass.

The paper is organised around one identity and its consequences. After the
rescaling $\Tk F(z)=F(z/\sqrt{2k})$, the norm of $\Ak$ and the norm of $\Fock$
are two quadratic forms on the same space of power series, differing by the
single sequence $\wk(m)=(2k)^m/(2k)_m$ of \cref{prop:identity}. That sequence
takes values in $(0,1]$, equals one on the two lowest levels, decreases in $m$
and increases in $k$ to one. Its monotonicity in $k$ replaces every estimate
that would otherwise have to be uniform.

\Cref{sec:setting} sets up the two spaces and establishes that identity,
and shows that the scale $\sqrt{2k}$ is forced: it is the only rescaling whose
limiting weight is a nondegenerate Gaussian, and centring at the origin costs
no generality because the whole setting is $SU(1,1)$ covariant.
\Cref{sec:collapse} proves that the rescaled family is normal, that every
nonzero limit point lies in the Segal--Bargmann space, and that those with
finitely many zeros have the form $P\rme^{az^2+bz+c}$ with $\abs{a}<1/2$, so
that they are Fock states of finite stellar rank. \Cref{ex:cat} shows that the
finiteness hypothesis cannot be dropped. \Cref{sec:rank} shows that the rank is
upper semicontinuous and is preserved exactly when every zero is
$O(1/\sqrt{2k})$, the scale that also carries the coherent-state labels, and
that for unbounded rank the statement survives as convergence of the zero
counting measures.

\Cref{sec:rigidity} turns to degeneration, where a two-sided Harnack estimate
pins the modulus of the rescaled function on compact sets by a single number,
the geometric mean on a circle, once a Blaschke factor has removed the zeros.
By a separate argument that does not use that estimate, a family of rank $r$
degenerates if and only if the amplitudes on the lowest $r+1$ weight vectors
tend to zero. At $r=0$ this is the vanishing of the vacuum
fidelity. \Cref{sec:escape} shows that mass can be lost without the limit
vanishing: there is a family of unit norm whose limit has norm $2^{-1/2}$, so that half
of its mass is not accounted for by any compact set. The transition is
governed by a semiclassical potential on the disk with a degenerate double
well at an explicit coupling, and the same computation explains why only the
exponents $n\le2$ survive the contraction. Finally, \cref{sec:limitmap}
determines how much of a family the limit remembers. The answer is the lowest
$2r+3$ amplitudes and no fewer, and every admissible limit is attained by a
section that is holomorphic in the parameters.

Everything above is proved. The one point at which we stop short is the
behaviour exactly at the critical coupling $\lambda^*=0.9784$ of the escape
transition, where the retained fraction is a definite number that we determine
numerically but do not derive.
Throughout, states are pure and single-mode, the representation is a single
irreducible discrete-series representation of $SU(1,1)$, and $k>1/2$. The case
$k=1/2$ is the Hardy space and is noted separately where it differs. Mixed
states, several modes, and the continuous and principal series lie outside the
scope.

\section{The two spaces and the contraction}
\label{sec:setting}

This section fixes the two Hilbert spaces, the map between them, and the one
sequence that measures their difference. Nothing here is new, with the
exception of the way we arrange it, but the arrangement decides what the rest
of the paper has to prove, so we set it out in full.

\subsection{Bergman and Fock spaces}

\begin{definition}\label{def:spaces}
For $k>1/2$ let $\Ak$ denote the space of holomorphic functions
$F(\zeta)=\sum_{m\ge0}c_m\zeta^m$ on the unit disk $\D$ with
\begin{equation}\label{eq:Aknorm}
  \norm{F}_k^2 \;=\; \sum_{m\ge0}\abs{c_m}^2\,\frac{m!}{(2k)_m}\;<\;\infty ,
  \qquad (2k)_m:=\frac{\Gamma(2k+m)}{\Gamma(2k)} .
\end{equation}
This is a weighted Bergman space, with reproducing kernel
$K_k(\zeta,\zeta')=(1-\overline{\zeta'}\zeta)^{-2k}$ and weight
$(1-\abs{\zeta}^2)^{2k-2}$. Let $\Fock$ denote the Segal--Bargmann space of
entire functions $G(z)=\sum_{m\ge0}b_mz^m$ with
$\norm{G}_\Fock^2=\sum_m\abs{b_m}^2m!$ and kernel $\rme^{\overline{z'}z}$.
\end{definition}

The space $\Ak$ carries the discrete-series representation of $SU(1,1)$ with
Bargmann index $k$, whose coherent states in the sense of Perelomov
\cite{Perelomov1972} are the orbit of the lowest weight vector. For the
construction of coherent states from a dynamical group, and for the
coset-space picture that underlies it, we refer to \cite{Gilmore1972,ZFG1990} and to the
textbook treatment in \cite{KZF2023}. The space $\Fock$ carries the
Schr\"odinger representation of the Heisenberg--Weyl group $H_1$. A state
$\ket\psi=\sum_m\psi_m\ket{k,k+m}$ of the discrete series corresponds to the
function $F$ whose coefficients are read off \eqref{eq:Aknorm}, namely
\begin{equation}\label{eq:psic}
  \psi_m=c_m\Bigl[\frac{m!}{(2k)_m}\Bigr]^{1/2},
\end{equation}
so that $\norm{F}_k$ is the norm of the state.

Three sets of coefficients appear throughout and are easy to confuse, so we
fix them now and use them consistently. The letters $c_m$ always denote Taylor
coefficients on the disk, $b_m$ always denote Taylor coefficients on the plane
after the rescaling of \cref{def:contraction}, and $\psi_m$ always denote
amplitudes, that is, coefficients in an orthonormal basis. The disk variable is
$\zeta$ and the plane variable is $z$.

The Husimi function of a state is
$Q_F(\zeta)=(1-\abs{\zeta}^2)^{2k}\abs{F(\zeta)}^2$, and the $SU(1,1)$
invariant measure on $\D$ is
$\rmd\nu_k=\tfrac{2k-1}{\pi}(1-\abs{\zeta}^2)^{-2}\dA$.

We call $F$ the \emph{stellar function} of the state, and the \emph{stellar
rank} $r(F)$ is the number of zeros of $F$ in its domain of holomorphy,
counted with multiplicity. Since $Q_F$ and $F$ have the same zeros, this agrees with the definition
through the Husimi function used in \cite{Chabaud2020} on the plane and in
\cite{Klimov2026} on the disk. On the plane a state has rank zero if and only
if it is Gaussian \cite{Chabaud2020}, whereas the corresponding statement on the disk
fails, which is the point of departure of this paper.

\begin{lemma}\label{lem:kernelbound}
For every $F\in\Ak$ and every $\zeta\in\D$,
\begin{equation}\label{eq:covbound}
  (1-\abs{\zeta}^2)^{k}\,\abs{F(\zeta)} \;\le\; \norm{F}_k .
\end{equation}
\end{lemma}

\begin{proof}
Evaluation at $\zeta$ is $F\mapsto\ip{F}{K_k(\cdot,\zeta)}$, so
Cauchy--Schwarz gives
\[
  \abs{F(\zeta)}\le\norm{F}_k\,K_k(\zeta,\zeta)^{1/2}
  =\norm{F}_k\,(1-\abs{\zeta}^2)^{-k} . \qedhere
\]
\end{proof}

We record \eqref{eq:covbound} in this form rather than as a bound on
$\abs{F(\zeta)}$ alone because the left-hand side is $Q_F^{1/2}$, which is
$SU(1,1)$ covariant. \Cref{lem:cocycle} below makes the covariance precise,
and \cref{prop:scale} and the whole of \cref{sec:rigidity} rest on it.

\subsection{The contraction and the weight sequence}

A contraction of Lie groups in the sense of In\"on\"u and Wigner
\cite{InonuWigner1953} deforms one group into another of the same dimension, and the associated deformation of
representations turns the discrete series of $SU(1,1)$ into the
Schr\"odinger representation of $H_1$. At the level of generators the
statement is that $(2k)^{-1/2}K_-\to a$ and $(2k)^{-1/2}K_+\to a^\dagger$ as
$k\to\infty$. See Kastrup \cite[section 2.4]{Kastrup2005}, who traces the
scaling to Barut and Girardello \cite{BG1971}. On functions this amounts to a
single rescaling of the variable.

\begin{definition}\label{def:contraction}
For $F\in\Ak$ set $(\Tk F)(z):=F(z/\sqrt{2k})$, a function holomorphic on
$\abs{z}<\sqrt{2k}$. We write $\Ft_k:=\Tk F_k$ and denote its Taylor
coefficients by $b^{(k)}_m$, suppressing the superscript when no confusion can
arise.
\end{definition}

The rescaling sends the Bergman weight to the Fock weight and the Bergman
kernel to the Fock kernel, both by the elementary limit
$(1-x/n)^{n}\to\rme^{-x}$ applied with $n=2k$,
\begin{equation}\label{eq:weightkernel}
  \Bigl(1-\frac{\abs{z}^2}{2k}\Bigr)^{2k}\longrightarrow\rme^{-\abs{z}^2},
  \qquad
  \Bigl(1-\frac{\overline{z'}z}{2k}\Bigr)^{-2k}\longrightarrow
  \rme^{\overline{z'}z},
\end{equation}
and the disk $\abs{\zeta}<1$ becomes the disk $\abs{z}<\sqrt{2k}$, which
exhausts $\C$. The norm transforms as follows.

\begin{proposition}\label{prop:identity}
With $b_m$ the Taylor coefficients of $\Tk F$,
\begin{equation}\label{eq:key}
  \norm{F}_k^2 \;=\; \sum_{m\ge0}\abs{b_m}^2\,m!\,\wk(m),
  \qquad
  \wk(m):=\frac{(2k)^m}{(2k)_m}=\prod_{j=0}^{m-1}\Bigl(1+\frac{j}{2k}\Bigr)^{-1},
\end{equation}
and the amplitudes satisfy $\psi_m=b_m\sqrt{m!\,\wk(m)}$.
\end{proposition}

\begin{proof}
$b_m=(2k)^{-m/2}c_m$. Substitute into \eqref{eq:Aknorm} and use
$\wk(m)=(2k)^m/(2k)_m$. The second identity then follows from
\eqref{eq:psic}. The details are in \cref{app:A2}.
\end{proof}

It is tempting to read \eqref{eq:key} as a comparison of the two norms, and
this is nearly right, but $\Tk F$ need not lie in $\Fock$: for $F\in\Ak$ it is
holomorphic only on $\abs{z}<\sqrt{2k}$, and even when it continues to an
entire function its Fock norm may be infinite. The accurate statement is that
the two norms are quadratic forms on one and the same space of formal power
series. Writing
\begin{equation}\label{eq:twoforms}
  Q_k[b]:=\sum_m\abs{b_m}^2m!\,\wk(m),
  \qquad
  Q_\infty[b]:=\sum_m\abs{b_m}^2m! ,
\end{equation}
we have $\norm{F}_k^2=Q_k[b]$, while finiteness of $Q_\infty[b]$ characterises
membership of $\Fock$. The whole difference between the two spaces is carried
by the sequence $\wk$, whose properties we now list.

\begin{lemma}\label{lem:weight}
For every $k>0$ the sequence $\wk$ satisfies
\begin{enumerate}
\item[(i)] $0<\wk(m)\le1$, with $\wk(0)=\wk(1)=1$ and $\wk(m)<1$ for $m\ge2$;
\item[(ii)] $\wk(m)$ is nonincreasing in $m$, and strictly decreasing for
  $m\ge1$;
\item[(iii)] $\wk(m)$ is strictly increasing in $k$ for each $m\ge2$;
\item[(iv)] $\wk(m)\uparrow1$ as $k\to\infty$, for each fixed $m$.
\end{enumerate}
\end{lemma}

\begin{proof}
Each assertion reduces to a statement about a single factor of the product in
\eqref{eq:key}. The proofs, and a quantitative form of the decay in $m$, are
in \cref{app:weight}.
\end{proof}

In the language of \eqref{eq:twoforms}, \cref{lem:weight} says that
$Q_k\le Q_\infty$ for every $k$, that the two forms agree on the span of the
two lowest levels, and that $Q_k\uparrow Q_\infty$ pointwise. The contraction
is a monotone deformation of one quadratic form into another on a fixed space,
proceeding one coefficient at a time.

Property (iii) is what makes the estimates below soft. Wherever a bound
uniform in $k$ would otherwise be required, monotonicity in $k$ allows us to compare
$\wk$ with $w_{k_0}$ at a fixed $k_0\le k$ and then pass to the limit by
monotone convergence. \Cref{prop:fatou} is the clearest instance, and
\cref{thm:section} uses the same mechanism in the opposite direction.

\subsection{Covariance and the uniqueness of the scale}

The rescaling of \cref{def:contraction} is centred at the origin, and the
origin has no distinguished role in $\D$, since $SU(1,1)$ acts transitively.
Before anything else we therefore check that nothing is lost by centring, and
that the power $\sqrt{2k}$ is forced rather than chosen.

\begin{lemma}\label{lem:cocycle}
Let $g=\begin{psmallmatrix}a&b\\ \bar b&\bar a\end{psmallmatrix}$ with
$\abs{a}^2-\abs{b}^2=1$ act on $\D$ by
$g\cdot\zeta=(a\zeta+b)/(\bar b\zeta+\bar a)$. Then
\begin{equation}\label{eq:cocycle}
  1-\abs{g\cdot\zeta}^2=\frac{1-\abs{\zeta}^2}{\abs{\bar b\zeta+\bar a}^2}.
\end{equation}
Consequently $(U_gF)(\zeta):=(\bar b\zeta+\bar a)^{-2k}F(g\cdot\zeta)$ defines
a unitary representation of $SU(1,1)$ on $\Ak$, the Husimi amplitude
$(1-\abs{\zeta}^2)^k\abs{F(\zeta)}$ of \eqref{eq:covbound} is invariant
pointwise, in the sense that its value at $\zeta$ for $U_gF$ equals its value
at $g\cdot\zeta$ for $F$, and $U_g$ preserves the number of zeros in $\D$.
\end{lemma}

\begin{proof}
Expand both moduli and use $\abs{a}^2-\abs{b}^2=1$. See \cref{app:jensen}.
\end{proof}

\begin{proposition}\label{prop:scale}
Consider rescalings of the form $\zeta=g_k(\lambda_ku)$ with $g_k\in SU(1,1)$
and $\lambda_k\to0$. Replacing $F_k$ by $U_{g_k}^{-1}F_k$ changes neither the
norm, nor the number of zeros in $\D$, nor the Husimi amplitude, and reduces
the rescaling to the centred one. For the centred rescaling,
\begin{equation}\label{eq:weightlimit}
  (1-\lambda_k^2\abs{u}^2)^{2k}\longrightarrow
  \exp\Bigl(-\lim_k 2k\lambda_k^2\,\abs{u}^2\Bigr),
\end{equation}
and the limit weight is nondegenerate if and only if
$2k\lambda_k^2\to c\in(0,\infty)$, that is, $\lambda_k\asymp1/\sqrt{2k}$. If
$2k\lambda_k^2\to\infty$ the limit weight vanishes off the origin, and if
$2k\lambda_k^2\to0$ it is identically one and imposes no confinement.
\end{proposition}

The proof is in \cref{app:B1}. The reduction it performs is not an appeal to a
normalisation. The
hypotheses and the conclusions of every statement in this paper are invariant
under $U_g$, so a family and its image under a $k$-dependent group element
stand or fall together, and the centred case is the general one.

\subsection{The limit $k\to\infty$}

The Bargmann index is not a free parameter. It is fixed by the realisation of
the algebra, and only some realisations admit large $k$.
\Cref{tab:realizations} lists three. The third gives the cleanest physical
reading, in which $k\to\infty$ is a large-angular-momentum limit and is the
non-compact counterpart of the large-spin limit $j\to\infty$ for $SU(2)$. In
the two-mode realisation the same limit is one of large photon-number
imbalance between the modes, and that realisation is the one underlying the
$SU(1,1)$ interferometers of \cite{YMK1986}.

\begin{table}[t]
\centering
\begin{tabular}{lll}
\toprule
Realisation & $k$ & meaning of $k\to\infty$ \\
\midrule
single mode, $K_+=\tfrac12(a^\dagger)^2$ & $1/4$, $3/4$
  & not available${}^{\dagger}$ \\
two modes, $K_+=a_1^\dagger a_2^\dagger$ & $\tfrac12(\abs{n_1-n_2}+1)$
  & large photon-number imbalance \\
radial part, isotropic oscillator & $\tfrac12(\ell+3/2)$
  & large angular momentum \\
\bottomrule
\end{tabular}
\caption{Three realisations of the discrete series, and what the contraction
limit means in each. ${}^{\dagger}$The first row lies outside the standing assumption $k>1/2$ and
is listed only to record that the contraction is unavailable there.}
\label{tab:realizations}
\end{table}

The first row explains why the contraction is unavailable for single-mode
squeezing. There $K_+=\tfrac12(a^\dagger)^2$ gives Casimir $-3/16$ and acts within the
even and odd photon-number sectors separately, so the Bargmann index is pinned
at $1/4$ or $3/4$ by the algebra itself \cite[section 5]{BVM1996}, with no
parameter left to send to infinity. The
contraction is a statement about a family of representations, and the
single-mode realisation supplies only two of them.

The second row is the one to keep in mind. There the sectors are labelled by
the photon-number difference $m=n_1-n_2$, which is conserved, and the
Perelomov coherent states are the two-mode squeezed states. The two-mode
squeezed vacuum itself sits at $m=0$, that is $k=1/2$, the Hardy case. Letting
$k\to\infty$ means preparing the two modes with a large fixed imbalance, and
the contraction then says that a two-mode squeezed state whose parameter is
scaled as $\zeta_0=\alpha/\sqrt{2k}$ becomes indistinguishable from a Glauber
coherent state with label $\alpha$. Weak squeezing on top of a large imbalance
looks like a displacement.

None of this is the Berezin semiclassical limit recalled in \cref{sec:intro},
which acts on symbols of operators, whereas the rescaling of
\cref{def:contraction} acts on states.

\begin{remark}[The Hardy case]\label{rem:hardy}
At $k=1/2$ the weight $(1-\abs{\zeta}^2)^{2k-2}$ is not integrable and the
factor $2k-1$ in $\rmd\nu_k$ vanishes, so the measure-theoretic side of the
setting degenerates. The resolution of the identity behind $\rmd\nu_k$ holds
for $k>1/2$, with the limit $k\to1/2$ to be taken after the angular
integration \cite[section 2.1]{BVM1996}; an analytic continuation gives a
weak resolution valid for $k<1/2$ as well \cite[section 4]{BVM1996}, which we
do not need here. What survives is everything that uses only the monomial
norms and the reproducing kernel. There $(2k)_m=m!$, so \eqref{eq:Aknorm}
reads $\norm{F}^2_{1/2}=\sum_m\abs{c_m}^2$ and $\Ak$ is the Hardy space
$H^2(\D)$ with kernel $(1-\overline{\zeta'}\zeta)^{-1}$, while the weight
sequence $\wk$ of \eqref{eq:key} is defined for every $k>0$ and
\cref{lem:weight} holds unchanged. Consequently
\cref{sec:collapse,sec:rank}, the rigidity of
\cref{lem:harnack,thm:harnack-r} and the whole of \cref{sec:limitmap} apply
verbatim at $k=1/2$, as does the statement that the zero-free class is
infinite-dimensional, since $H^\infty(\D)\subset H^2(\D)$.

The covariance of \cref{lem:cocycle} also survives, although its proof does
not: the unitarity of $U_g$ is verified there by integrating against
$\rmd\nu_k$, which is unavailable at $k=1/2$. What replaces it is the boundary
form of the Hardy norm. For $2k=1$ the identity
$\abs{\bar b\zeta+\bar a}^{-2}\,\rmd\theta=\rmd\theta'$ on $\abs{\zeta}=1$,
where $\theta'$ is the argument of $g\cdot\zeta$, gives
$\frac1{2\pi}\int\abs{(U_gF)(\rme^{\rmi\theta})}^2\rmd\theta
=\frac1{2\pi}\int\abs{F}^2\rmd\theta'$ directly, so $U_g$ is unitary and
\cref{prop:scale} keeps its centring reduction.

Two things do not carry over. \Cref{lem:beta} needs $2k-1>0$, so the Husimi
measure $\mu_k$, and with it parts (iii) and (iv) of \cref{thm:escape} and all
of \cref{sec:escape}, require $k>1/2$. And the pathology of Bergman zero sets
invoked after \cref{lem:harnack} is absent, since Hardy-space zero sets do
satisfy the Blaschke condition, although the rigidity argument uses neither
fact. The value $k=1/2$ is also where the random model becomes determinantal,
as recalled in \cref{sec:intro}.
\end{remark}

\section{The collapse theorem}
\label{sec:collapse}

The contraction rescales a family of states, one for each $k$, into a family
of functions on larger and larger disks. This section establishes that the
family has limit points, that the limits lie in $\Fock$, and that they are
finite-rank Fock states whenever they have finitely many zeros. The three
steps are independent and use different properties of $\wk$ from
\cref{lem:weight}.

\subsection{Normality}

The rescaled functions live on disks $\abs{z}<\sqrt{2k}$ that exhaust $\C$, so
a limit can be discussed only on compact sets, and only after a bound that
does not deteriorate with $k$. The reproducing kernel supplies such a bound.

\begin{lemma}[Normality]\label{lem:normal}
Let $\norm{F_k}_k\le1$. Then
\begin{equation}\label{eq:rescaledbound}
  \abs{\Ft_k(z)}\le\Bigl(1-\frac{\abs{z}^2}{2k}\Bigr)^{-k}
  \qquad\text{for }\abs{z}^2<2k,
\end{equation}
and the right-hand side is strictly decreasing in $k$ with limit
$\rme^{\abs{z}^2/2}$. Consequently, for every $R>0$ and every $k_0>R^2/2$,
\begin{equation}\label{eq:M}
  \sup_{\abs{z}\le R}\abs{\Ft_k(z)}\ \le\ M
  :=\Bigl(1-\frac{R^2}{2k_0}\Bigr)^{-k_0}
  \qquad\text{for all }k\ge k_0 .
\end{equation}
Consequently $\{\Ft_k\}_{k\ge k_0}$ is uniformly bounded on compact subsets of
$\abs{z}<\sqrt{2k_0}$, and by a diagonal argument over an exhaustion of $\C$ it
is a normal family on $\C$.
\end{lemma}

\begin{proof}
The bound \eqref{eq:rescaledbound} is \eqref{eq:covbound} at
$\zeta=z/\sqrt{2k}$. For the monotonicity, fix $t=\abs{z}^2>0$ and put
$u=t/2k\in(0,1)$, so that
\[
  -k\log\Bigl(1-\frac{t}{2k}\Bigr)=\frac{t}{2}\cdot\frac{-\log(1-u)}{u}
  =\frac{t}{2}\sum_{n\ge1}\frac{u^{n-1}}{n} .
\]
The series has positive coefficients, so it is strictly increasing in $u$,
and $u$ decreases as $k$ increases. Expanding, the right-hand side is
$\tfrac{t}{2}+\tfrac{t^2}{8k}+O(k^{-2})$, so the limit $u\to0$ gives $t/2$ and
the approach is from above at rate $1/k$. Since
\eqref{eq:rescaledbound} holds for every $k$ and its right-hand side
decreases, the value at $k_0$ bounds the whole tail, which is \eqref{eq:M}.
Montel's theorem then gives normality.
\end{proof}

From here on we work with limit points along subsequences. Whether the family
converges as a whole is a separate question, settled in \cref{thm:unique}.

The mechanism of \cref{lem:normal}, a norm bound converted into a locally
uniform bound and then into normality, is the one Ricci uses in the compact
case $SU(2)\to H_1$ \cite[Proposition 1]{Ricci1986}, where the rescaled
objects are the Maclaurin sections of a function in $\Fock$ and the bound
comes from their Fock norms. The zeros play no role there, and the question
of what happens to them was not available at the time.

\subsection{Lower semicontinuity of the Fock norm}

A limit point of a norm-bounded family need not have the same norm, and in
\cref{sec:escape} we exhibit families whose limits have strictly smaller norm,
the deficit being mass that escapes to infinity under the rescaling. What does
survive is an inequality, and it is here that monotonicity of $\wk$ in $k$
enters for the first time.

\begin{proposition}\label{prop:fatou}
Suppose $\norm{F_k}_k\le1$ and $\Ft_k\to\Finf$ locally uniformly along a
subsequence. Then $\Finf\in\Fock$ and
\begin{equation}\label{eq:lsc}
  \norm{\Finf}_\Fock\ \le\ \liminf_k\ \norm{F_k}_k .
\end{equation}
\end{proposition}

\begin{proof}
Local uniform convergence gives $b^{(k)}_m\to b_m$ for each fixed $m$, by the
Cauchy integral formula on a fixed circle. Fix $k_0>1/2$ and $N\in\N$. For
every $k\ge k_0$, \cref{lem:weight}(iii) gives $w_{k_0}(m)\le\wk(m)$, so
\[
  \sum_{m\le N}\abs{b^{(k)}_m}^2m!\,w_{k_0}(m)
  \ \le\ \sum_{m\le N}\abs{b^{(k)}_m}^2m!\,\wk(m)
  \ \le\ \sum_{m\ge0}\abs{b^{(k)}_m}^2m!\,\wk(m)
  \ =\ \norm{F_k}_k^2 ,
\]
the last step being \eqref{eq:key}. The left-hand side is a finite sum, so
letting $k\to\infty$ along the subsequence,
\[
  \sum_{m\le N}\abs{b_m}^2m!\,w_{k_0}(m)\ \le\ \liminf_k\norm{F_k}_k^2 .
\]
The right-hand side does not involve $k_0$ or $N$. Letting $k_0\to\infty$ and
using $w_{k_0}(m)\uparrow1$ from \cref{lem:weight}(iv), then letting
$N\to\infty$, monotone convergence gives
$\sum_{m\ge0}\abs{b_m}^2m!\le\liminf_k\norm{F_k}_k^2$, which is finite. Hence
$\Finf\in\Fock$ and \eqref{eq:lsc} holds.
\end{proof}

The proof uses no estimate that is uniform in $k$. Both limits are monotone,
one by \cref{lem:weight}(iii) and one by \cref{lem:weight}(iv), and the
argument is Fatou's lemma applied to a doubly indexed family. This is the
pattern referred to after \cref{lem:weight}, and it recurs in
\cref{thm:section} with the roles of the two limits exchanged.

\subsection{Structure of the limit points}

The limits can now be identified, using two classical facts. Hurwitz's theorem
controls how zeros behave in the limit, and the
Borel--Carath\'eodory inequality converts a bound on the real part of an
entire function into a bound on its modulus, which is what turns the growth
estimate \eqref{eq:rescaledbound} into a statement about degree.

\begin{theorem}[Collapse]\label{thm:collapse}
Let $F_k\in\Ak$ with $\norm{F_k}_k\le1$. Then $\{\Ft_k\}$ is a normal family on
$\C$, and every locally uniform limit point $\Finf$ satisfies the following.
\begin{enumerate}
\item[(i)] Either $\Finf\equiv0$, or $\Finf\in\Fock$ with
  $\norm{\Finf}_\Fock\le\liminf_k\norm{F_k}_k$.
\item[(ii)] If $\Finf\not\equiv0$ has finitely many zeros, then
\begin{equation}\label{eq:limitform}
  \Finf(z)=P(z)\,\rme^{az^2+bz+c},
  \qquad \deg P=r(\Finf),\quad \abs{a}<\tfrac12 ,
\end{equation}
that is, $\Finf$ is a Fock-space state of stellar rank $\deg P$.
\end{enumerate}
\end{theorem}

\begin{proof}
Normality is \cref{lem:normal}, and (i) is \cref{prop:fatou}.

For (ii), let $d=r(\Finf)<\infty$ and factor $\Finf=PH$ with $P$ a polynomial
of degree $d$ carrying the zeros and $H$ entire and zero-free. Since $\C$ is
simply connected, $H=\rme^{g}$ for some entire $g$.

Passing to the limit in \eqref{eq:rescaledbound} gives
$\abs{\Finf(z)}\le\rme^{\abs{z}^2/2}$ for all $z$. Choose $R_0$ larger than the
modulus of every zero of $P$, so that $\abs{P(z)}\ge c\abs{z}^{d}$ on
$\abs{z}\ge R_0$ for some $c>0$, so
\[
  \Ree\,g(z)=\log\abs{H(z)}\le\frac{\abs{z}^2}{2}-d\log\abs{z}-\log c
  \le\frac{\abs{z}^2}{2}
\]
for $\abs{z}$ large. The Borel--Carath\'eodory inequality \cite[section 5.5]{Titchmarsh1939} on
the circle $\abs{z}=t$ with outer radius $2t$ gives
$M_g(t)\le 2\cdot\tfrac12(2t)^2+3\abs{g(0)}=4t^2+3\abs{g(0)}$, so
$M_g(t)=O(t^2)$, and the Cauchy estimates
$\abs{g_n}\le M_g(t)t^{-n}$ force $g_n=0$ for $n\ge3$. Hence
$g(z)=az^2+bz+c$.

It remains to see that $\abs{a}<\tfrac12$. Writing $z=\rho\rme^{\rmi\theta}$
and choosing $\theta$ so that $az^2$ is positive real, the integrand of
$\norm{\Finf}^2_\Fock$ behaves as
$\abs{P}^2\rme^{(2\abs{a}-1)\rho^2}$ along that ray, which is not integrable
unless $\abs{a}\le\tfrac12$. At $\abs{a}=\tfrac12$ the exponent vanishes on
the ray and the Fock norm still diverges: for $\Finf=\rme^{z^2/2}$ the
coefficients give
$\norm{\Finf}^2_\Fock=\sum_j\binom{2j}{j}4^{-j}$, whose terms are asymptotic
to $(\pi j)^{-1/2}$, and a polynomial prefactor does not change this. Hence
$\abs{a}<\tfrac12$, and \eqref{eq:limitform} is exactly the normal form of a
Fock state of stellar rank $\deg P$ \cite{Chabaud2020}.
\end{proof}

The Borel--Carath\'eodory step cannot be omitted. Absence of zeros alone
yields $H=\rme^{g}$ with $g$ entire and nothing more, and transcendental $g$ of
small order has to be excluded. Only a bound on $\abs{g}$ does that, and the
growth estimate delivers a bound on $\Ree\,g$, so the two must be connected.

The strictness of $\abs{a}<\tfrac12$ in \eqref{eq:limitform} is a fact about
the Fock weight and not a general feature. Replacing the weight $h_n=n!$ by
$h_n=n!/(n+1)^3$ leaves the order and the type of the space unchanged, but the
terms of $\norm{\rme^{z^2/2}}^2$ then decay as $j^{-7/2}$ instead of
$j^{-1/2}$ and the series converges, so the corresponding constraint is
$\abs{a}\le\tfrac12$ with the boundary attained. Whether the inequality is
strict has to be decided weight by weight. The same question arises for the
Barut--Girardello weight in \cref{sec:discussion}, where the constraint reads
$\abs{b}<1$ and is again strict.

One consequence of \cref{lem:weight}(i) belongs here as well, because it fixes
the scale that governs the whole of \cref{sec:rank}.

\begin{proposition}\label{prop:exact2}
For every $k>1/2$ the map $\Tk$ restricted to
$\operatorname{span}\{1,\zeta\}\subset\Ak$ is an isometry onto
$\operatorname{span}\{1,z\}\subset\Fock$, and this two-dimensional subspace is
the largest on which $\Tk$ is isometric.
\end{proposition}

\begin{proof}
See \cref{app:A3}. Both assertions follow from $\wk(0)=\wk(1)=1$ together with
$\wk(m)<1$ for $m\ge2$.
\end{proof}

Concretely, take the rank-one state with stellar function
$F_k(\zeta)=\zeta-\alpha/\sqrt{2k}$, whose single zero sits at the scale that
\cref{thm:rank-eq} will identify as the critical one. Its amplitudes are
$\psi_0=-\alpha/\sqrt{2k}$ and $\psi_1=1/\sqrt{2k}$, so after normalisation
they are proportional to $(-\alpha,1)$, and the Fock state with stellar
function $z-\alpha$ has amplitudes $(-\alpha,1)$ as well. The two states
coincide at every $k$, and not asymptotically. In this sense the scale
$1/\sqrt{2k}$ for zeros is an identity at the bottom of the ladder, and
\cref{thm:rank-eq} extends it upward in the only form available there, as a
statement about limits.

Exactness fails first at level two, and it fails at the rate
$1-\wk(2)=(2k+1)^{-1}$. Every rate of convergence in $k$ that appears below
traces back to this.

\section{Rank under contraction}
\label{sec:rank}

\Cref{thm:collapse} identifies the limits but says nothing about which of them
occur, and in particular nothing about how the rank of $F_k$ relates to the
rank of the limit. That relation is the subject of this section, and it turns
on a single scale.

Throughout, $r_k$ denotes the number of zeros of $F_k$ in $\D$ counted with
multiplicity, and $\zeta^{(k)}_1,\dots,\zeta^{(k)}_{r_k}$ denote those zeros.
Since $\Ft_k(z)=F_k(z/\sqrt{2k})$, the zeros of $\Ft_k$ sit at
\begin{equation}\label{eq:zeroscale}
  z^{(k)}_j=\sqrt{2k}\,\zeta^{(k)}_j .
\end{equation}
A zero anchored at a fixed point of $\D$ therefore recedes to infinity in the
rescaled picture, at rate $\sqrt{2k}$, while a zero that itself shrinks like
$1/\sqrt{2k}$ stays where it is. The rank of the limit counts the second kind
and misses the first.

\subsection{Semicontinuity and the surviving scale}

\begin{theorem}[Semicontinuity]\label{thm:semicont}
Let $\norm{F_k}_k\le1$ and let $\Finf\not\equiv0$ be a limit point of
$\{\Ft_k\}$. Then
\begin{equation}
  r(\Finf)\ \le\ \liminf_k\, r_k .
\end{equation}
\end{theorem}

\begin{proof}
Since $\Finf\not\equiv0$ its zeros are isolated, so all but countably many
$R>0$ satisfy $\Finf\ne0$ on $\abs{z}=R$. Fix such an $R$. For $k>R^2/2$ the
disk $\abs{\zeta}<R/\sqrt{2k}$ lies inside $\D$, so by \eqref{eq:zeroscale} the
number of zeros of $\Ft_k$ in $\abs{z}<R$ is at most $r_k$. Hurwitz's theorem,
applied on $\abs{z}<R$ along the convergent subsequence, gives that $\Finf$ and
$\Ft_k$ have equally many zeros there for large $k$. Hence the number of zeros
of $\Finf$ in $\abs{z}<R$ is at most $\liminf_k r_k$, and letting $R\to\infty$
through admissible values completes the proof.
\end{proof}

The inequality can be strict, and the next result says exactly when. We first
record the mechanism, which is Hurwitz's theorem localised at each zero.

\begin{lemma}\label{lem:zeroconv}
Let $\Ft_k\to\Finf\not\equiv0$ locally uniformly, let $w$ be a zero of $\Finf$
of multiplicity $\nu$, and let $\varepsilon>0$ be small enough that
$\Finf$ has no other zero in $\abs{z-w}\le\varepsilon$. Then for all large $k$
the function $\Ft_k$ has exactly $\nu$ zeros in $\abs{z-w}<\varepsilon$. If in
addition $K\subset\C$ is compact and $\Finf$ has no zero in $K$, then $\Ft_k$
has no zero in $K$ for all large $k$.
\end{lemma}

\begin{theorem}[Equality]\label{thm:rank-eq}
Let $\norm{F_k}_k\le1$, suppose $r_k\equiv r<\infty$, and let
$\Ft_k\to\Finf\not\equiv0$ along a subsequence. Then
\begin{enumerate}
\item[(i)] $r(\Finf)=r$ if and only if
  $\sup_k\max_{j\le r}\sqrt{2k}\,\abs{\zeta^{(k)}_j}<\infty$, that is, if and
  only if every zero of $F_k$ is $O(1/\sqrt{2k})$;
\item[(ii)] in general, after passing to a further subsequence along which
  each $z^{(k)}_j$ either converges in $\C$ or tends to infinity, which is
  possible because $r$ is finite, $r(\Finf)$ equals the number of indices $j$
  in the first class, counted with multiplicity.
\end{enumerate}
\end{theorem}

\begin{proof}
For (ii), let $s$ be the number of indices whose rescaled zeros converge, with
limits $w_1,\dots,w_s$ listed with multiplicity, and choose $R$ larger than
$\max_i\abs{w_i}$ with $\Finf\ne0$ on $\abs{z}=R$. For large $k$ the remaining
$r-s$ rescaled zeros lie outside $\abs{z}\le R$, so $\Ft_k$ has exactly $s$
zeros in $\abs{z}<R$, and \cref{lem:zeroconv} gives that $\Finf$ has $s$ zeros
there. The same argument applies to every larger admissible radius, and the
count does not increase, so $\Finf$ has exactly $s$ zeros in $\C$.

For (i), if all rescaled zeros are bounded then $s=r$ in (ii) and
$r(\Finf)=r$. Conversely, if some $\abs{z^{(k)}_{j_0}}$ is unbounded, pass to
a further subsequence along which it tends to infinity. Then for every
admissible $R$ the function $\Ft_k$ has at most $r-1$ zeros in $\abs{z}<R$ for
large $k$, so by \cref{lem:zeroconv} the same holds for $\Finf$, and letting
$R\to\infty$ gives $r(\Finf)\le r-1$.
\end{proof}

Since $r(\Finf)$ does not depend on the further subsequence in (ii), neither
does the number of bounded rescaled zeros, although this is not obvious from
the definition.

The scale in \cref{thm:rank-eq} is the same one that labels the coherent
states, since the Perelomov state $\ket{k,\zeta_0}$ contracts to the Glauber state
$\ket\alpha$ precisely when $\zeta_0=\alpha/\sqrt{2k}$. A zero and a
coherent-state label are carried by the contraction in the same way, and
\cref{prop:exact2} shows that at rank one the correspondence is exact at every
$k$.

\subsection{Unbounded rank and counting measures}

The statement extends to families whose rank is unbounded, provided one
replaces the integer by the counting measure of the zeros. Weak-$*$
convergence of measures is the natural formulation because it is insensitive
to mass at infinity, which is exactly the mass that \cref{thm:rank-eq}
discards.

\begin{theorem}[Counting measures]\label{thm:measures}
Let $\Nz_k:=\sum_j\delta_{z^{(k)}_j}$ be the counting measure of the rescaled
zeros, a Radon measure on $\C$. If $\Ft_k\to\Finf\not\equiv0$ locally
uniformly, then
\begin{equation}
  \int_\C\varphi\,\rmd\Nz_k\ \longrightarrow\ \int_\C\varphi\,\rmd\Nz_{\Finf}
  \qquad\text{for every }\varphi\in C_c(\C),
\end{equation}
where $\Nz_{\Finf}$ is the counting measure of the zeros of $\Finf$.
\end{theorem}

The proof is \cref{lem:zeroconv} applied on small disks around the zeros of
$\Finf$ and on the complement of their union, with the details in
\cref{app:jensen}. We note that it uses neither the Poincar\'e--Lelong formula
nor $L^1_{\mathrm{loc}}$ convergence of $\log\abs{\Ft_k}$, which is the route
one would take if the zero sets were being handled through potential theory.
The elementary argument suffices because $\Finf\not\equiv0$ forces the zeros
to be locally finite.

In particular, when $r_k\equiv r<\infty$ and $\Finf\not\equiv0$,
\begin{equation}\label{eq:deficit}
  r-r(\Finf)\ =\ \lim_{R\to\infty}\ \lim_k\ \Nz_k\bigl(\{\abs{z}\ge R\}\bigr),
\end{equation}
so the rank deficit is exactly the mass of $\Nz_k$ that escapes to infinity.

The two theorems above concern families of fixed finite rank. The opposite
situation, in which the rank of $F_k$ grows without bound while the limit has
infinitely many zeros, is equally natural and is best seen on an example.

\begin{example}[Perelomov cat states]\label{ex:cat}
Take $2k\in\N$ and let
\begin{equation}
  \ket{\Psi_k}\propto\ket{k,\zeta_0}+\ket{k,-\zeta_0},
  \qquad
  F_k(\zeta)\propto(1-\zeta_0\zeta)^{-2k}+(1+\zeta_0\zeta)^{-2k} .
\end{equation}
The zeros of $F_k$ solve
$\bigl[(1+\zeta_0\zeta)/(1-\zeta_0\zeta)\bigr]^{2k}=-1$, that is
\begin{equation}\label{eq:catzeros}
  \zeta_n=\frac{\rmi}{\zeta_0}\,\tan\frac{\pi(2n+1)}{4k},
  \qquad n\in\Z ,
\end{equation}
of which only those with $\abs{\zeta_n}<1$ lie in the disk. The index runs
over a complete residue system modulo $2k$, and the zeros come in conjugate
pairs $\zeta_n,\zeta_{-1-n}$ on the imaginary axis. Since $2k$ is an
integer the equation has finitely many solutions there, so $r_k<\infty$ for
every $k$.

Now contract with $\zeta_0=\alpha/\sqrt{2k}$, the scale of
\cref{thm:rank-eq}. The rescaled zeros are
$z_n=\sqrt{2k}\,\zeta_n=\rmi\,(2k/\alpha)\tan[\pi(2n+1)/4k]$, and for fixed $n$ the argument of the tangent is $O(1/k)$, so
$\tan x=x+O(x^3)$ gives
\begin{equation}\label{eq:catlimit}
  z_n\ \longrightarrow\ \frac{\rmi\pi(2n+1)}{2\alpha},
\end{equation}
which are exactly the zeros of the Glauber cat
$\ket\alpha+\ket{-\alpha}$, whose stellar function is
$\cosh(\alpha z)$ up to a Gaussian factor. The condition $\abs{\zeta_n}<1$
reads $\tan[\pi(2n+1)/4k]<\abs{\alpha}/\sqrt{2k}$, and since the tangent is
increasing with argument $O(1/\sqrt{k})$ in the relevant range this is
$\pi(2n+1)/4k<\abs{\alpha}/\sqrt{2k}$ up to a relative error $O(1/k)$, that
is $n<\abs{\alpha}\sqrt{2k}/\pi-\tfrac12$. The number of admissible
$n$ therefore grows as $\abs{\alpha}\sqrt{2k}/\pi$, and counting the zeros in conjugate
pairs,
\begin{equation}\label{eq:catcount}
  r_k\ \sim\ \frac{2\sqrt2\,\abs{\alpha}}{\pi}\,\sqrt{k}\ \longrightarrow\ \infty .
\end{equation}
\Cref{tab:cat} confirms both statements numerically.
\end{example}

\begin{table}[t]
\centering
\begin{tabular}{rrlllll}
\toprule
$k$ & $r_k$ & $\Imm z_0$ & $\Imm z_1$ & $\Imm z_2$ & $\Imm z_3$
  & $r_k\big/\bigl[\tfrac{2\sqrt2\abs{\alpha}}{\pi}\sqrt{k}\bigr]$\\
\midrule
$20$      & $6$   & $1.047736$ & $3.156208$ & $5.304330$ & ---        & $0.9935$ \\
$100$     & $14$  & $1.047219$ & $3.142174$ & $5.238681$ & $7.337777$ & $1.0367$ \\
$10^3$    & $42$  & $1.047198$ & $3.141598$ & $5.236015$ & $7.330457$ & $0.9835$ \\
$10^4$    & $136$ & $1.047198$ & $3.141593$ & $5.235988$ & $7.330384$ & $1.0071$ \\
$10^5$    & $428$ & $1.047198$ & $3.141593$ & $5.235988$ & $7.330383$ & $1.0022$ \\
\midrule
limit     & $\infty$ & $1.047198$ & $3.141593$ & $5.235988$ & $7.330383$ & $1$ \\
\bottomrule
\end{tabular}
\caption{The Perelomov cat of \cref{ex:cat} at $\alpha=3/2$. The rescaled
zeros converge to $\pi(2n+1)/2\abs{\alpha}$ on the imaginary axis, which are
the zeros of the Glauber cat, and the number of zeros inside the disk grows
as predicted by \eqref{eq:catcount}.}
\label{tab:cat}
\end{table}

The example shows that \cref{thm:semicont} is not vacuous in the other
direction either. Here the rank of $F_k$ is finite for every $k$ and the rank
of the limit is infinite, so the inequality
$r(\Finf)\le\liminf_kr_k$ holds with both sides infinite and carries no
information. What survives is \cref{thm:measures}, which continues to describe
the zeros as measures.

One might hope that a zero at an intermediate scale survives under some
intermediate rescaling. It does not, because if
$\abs{\zeta^{(k)}_j}\asymp k^{-\gamma}$ with $0<\gamma<1/2$, the rescaled zero
sits at $\sqrt{2k}\,\abs{\zeta^{(k)}_j}\asymp k^{1/2-\gamma}$ and escapes by
\cref{thm:rank-eq}, and by \cref{prop:scale} no other rescaling recovers it,
because the limit weight is degenerate unless $2k\lambda_k^2$ has a finite
nonzero limit.

Nothing intermediate survives at the level of states either. For the rank-one
family $F_k(\zeta)=\zeta-\delta_k$ with
$1/\sqrt{2k}\ll\delta_k\ll1$ the amplitudes are $\psi_0=-\delta_k$ and
$\psi_1=1/\sqrt{2k}$, so after normalisation the state converges in the
Hilbert norm to the lowest weight vector, with no residual structure to
extract by any further rescaling.

\subsection{The contraction as a window}

The two mechanisms of this section, loss of rank in \cref{thm:rank-eq} and
gain of rank in \cref{ex:cat}, are one mechanism seen from two sides. In the
rescaled variable $z=\sqrt{2k}\,\zeta$ the image of the disk is
$\abs{z}<\sqrt{2k}$, a window that opens onto the plane as $k$ grows. A zero
anchored at a fixed point of $\D$ moves outward with the window edge and
leaves through it, while a zero anchored at a fixed point of the plane is overtaken
by the edge and enters. Which happens is decided by the size of a zero
relative to $\sqrt{2k}$, and by nothing else.

\Cref{fig:window} shows both mechanisms on a single family of rank two, one of
whose zeros is at the critical scale and one of which is not. In the disk the
two look alike; only after the rescaling does the difference appear.

The cat states of \cref{ex:cat} illustrate the same picture from the other
side. There the zeros of the limit form a fixed lattice on the imaginary axis
with spacing $\pi/\abs{\alpha}$, given by \eqref{eq:catlimit}, and the number
visible inside the window of radius $\sqrt{2k}$ is $2\abs{\alpha}\sqrt{2k}/\pi$,
which reproduces \eqref{eq:catcount}. The formula for $r_k$ is thus a
statement about how much of a fixed configuration the window admits, and not
about the state changing shape.

\begin{figure}[t]
\centering
\begin{tikzpicture}
  \begin{scope}[x=1.75cm,y=1.75cm]
    \draw[gray!50] (0,0) circle (1);
    \draw[->,gray!65] (-1.20,0) -- (1.28,0)
      node[right,black,font=\small] {$\Ree\zeta$};
    \draw[gray!65] (0,-1.15) -- (0,1.15);
    \draw[thick] (0.600,0) circle (0.042);
    \node[font=\scriptsize,above right,inner sep=1.5pt] at (0.615,0.02)
      {$\zeta_0$};
    \foreach \x in {0.4427,0.2214,0.0990} {\fill (\x,0) circle (0.036);}
    \draw[->,black!55,shorten >=1pt] (0.435,-0.20) -- (0.110,-0.20);
    \node[font=\scriptsize,below,inner sep=2pt] at (0.28,-0.21)
      {$\propto1/\sqrt{2k}$};
    \node[font=\small] at (0,-1.42) {zeros of $F_k$ in $\D$};
  \end{scope}
  \begin{scope}[xshift=5.4cm,x=0.42cm,y=0.42cm]
    \draw[->,gray!65] (-0.9,0) -- (10.8,0)
      node[right,black,font=\small] {$\Ree z$};
    \draw[gray!65] (0,-0.45) -- (0,0.45);
    \fill (1.400,0) circle (0.16);
    \node[font=\scriptsize,below,inner sep=3pt] at (1.40,-0.2)
      {$\alpha$, every $k$};
    \foreach \x/\lab in {1.897/5,3.795/20,8.485/100}
      {\draw[thick] (\x,0) circle (0.19);
       \node[font=\scriptsize,above,inner sep=1pt] at (\x,1.9) {$\lab$};
       \draw[gray!40] (\x,0.42) -- (\x,1.85);}
    \draw[->,black!55,shorten >=1pt] (2.55,-2.6) -- (8.05,-2.6);
    \node[font=\scriptsize,below,inner sep=2pt] at (5.3,-2.65)
      {$\propto\sqrt{2k}$};
    \node[font=\small] at (4.9,-4.3) {after $z=\sqrt{2k}\,\zeta$};
  \end{scope}
\end{tikzpicture}
\caption{The two mechanisms of \cref{thm:rank-eq}, on the rank-two family
$F_k(\zeta)=(\zeta-\alpha/\sqrt{2k})(\zeta-\zeta_0)$ with $\alpha=1.4$ and
$\zeta_0=0.6$. Filled marks are the zero that survives the contraction, open
marks the one that does not. Inside the disk the two behave differently, one
shrinking like $1/\sqrt{2k}$ and the other fixed, but nothing there says which
will survive. The rescaling $z=\sqrt{2k}\,\zeta$ separates them. The first is
pinned at $\alpha$ for every $k$, so its three positions coincide in a single
mark, while the second recedes at rate $\sqrt{2k}$ and is shown at $k=5$, $20$
and $100$. The limit has rank one. What decides the outcome is the
size of a zero relative to $\sqrt{2k}$, and nothing else.}
\label{fig:window}
\end{figure}
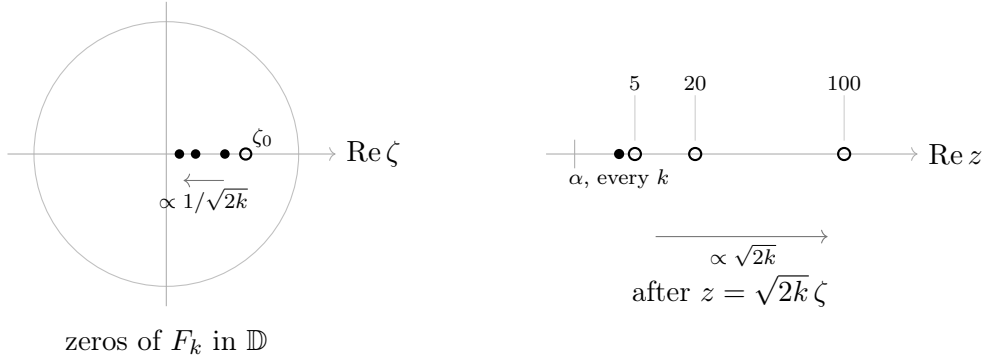

\section{Rigidity and the escape criterion}
\label{sec:rigidity}

\Cref{thm:collapse} leaves open the possibility that a family degenerates,
and \cref{sec:rank} says nothing about when. This section supplies the
criterion, in a form that is quantitative on compact sets. The mechanism is
that absence of zeros makes $\log\abs{\Ft_k}$ harmonic, and a nonnegative
harmonic function obeys a two-sided Harnack estimate. The whole profile on a
compact set is then pinned by a single number. Zeros are removed by a Blaschke
factor before the argument is applied, and the number that survives is the
geometric mean of $\abs{\Ft_k}$ on a circle.

\subsection{Rank zero}

\begin{lemma}[Harnack rigidity]\label{lem:harnack}
Let $F_k$ be zero-free with $\norm{F_k}_k\le1$, let $0<\rho<R$, and let $M$ be
as in \eqref{eq:M}. Then for $k\ge k_0$ and $\abs{z}\le\rho$,
\begin{equation}\label{eq:harnack}
  M\Bigl(\frac{\abs{F_k(0)}}{M}\Bigr)^{\frac{R+\rho}{R-\rho}}
  \ \le\ \abs{\Ft_k(z)}\ \le\
  M\Bigl(\frac{\abs{F_k(0)}}{M}\Bigr)^{\frac{R-\rho}{R+\rho}} .
\end{equation}
\end{lemma}

\begin{proof}
Since $\Ft_k$ is zero-free on the simply connected disk $\abs{z}<R$, we may
write $\Ft_k=\rme^{g_k}$ there, so $u_k:=\log\abs{\Ft_k}=\Ree\,g_k$ is
harmonic, where subharmonicity alone would be the general case. By
\eqref{eq:M} the function
$h_k:=\log M-u_k$ is nonnegative and harmonic on $\abs{z}<R$, and Harnack's
inequality gives
\[
  \frac{R-\rho}{R+\rho}\,h_k(0)\ \le\ h_k(z)\ \le\ \frac{R+\rho}{R-\rho}\,h_k(0)
  \qquad(\abs{z}\le\rho).
\]
Substituting $h_k(0)=\log\bigl(M/\abs{F_k(0)}\bigr)$, which uses
$\Ft_k(0)=F_k(0)$, and exponentiating gives \eqref{eq:harnack}.
\end{proof}

Only the supremum bound \eqref{eq:M} enters, and not the Bergman norm, so the
pathology of Bergman zero sets, which need not satisfy the Blaschke condition
\cite{Horowitz1974,Korenblum1975,HKZ2000}, plays no role here.
That is what makes the extension to finite rank a matter of dividing out a
factor rather than of controlling a norm.

\subsection{Finite rank}

\begin{lemma}[Local Blaschke factor]\label{lem:blaschke}
Let $Z_k$ denote the zeros of $\Ft_k$ in the \emph{open} disk $\abs{z}<R$,
listed with multiplicity, and set
\begin{equation}\label{eq:Bk}
  B_k(z):=\prod_{z_j\in Z_k}\frac{R(z-z_j)}{R^2-\bar z_j z},
  \qquad G_k:=\Ft_k/B_k .
\end{equation}
Then $\abs{B_k}\equiv1$ on $\abs{z}=R$ and $\abs{B_k}\le1$ inside, $G_k$ is
holomorphic and zero-free on $\abs{z}<R$, and $\abs{G_k}\le M$ there.
\end{lemma}

\begin{proof}
For $z=R\rme^{\rmi\theta}$,
\[
  \abs{R^2-\bar z_jz}=R\,\abs{R-\bar z_j\rme^{\rmi\theta}}
  =R\,\bigl|\overline{R\rme^{\rmi\theta}-z_j}\bigr|=R\abs{z-z_j},
\]
so each factor of $B_k$ has modulus one on $\abs{z}=R$, and $\abs{B_k}\le1$
inside by the maximum principle. The poles of $B_k$ are at $R^2/\bar z_j$,
of modulus $R^2/\abs{z_j}>R$, so $B_k$ is holomorphic on $\abs{z}<R$ with
exactly the zeros of $\Ft_k$ there. Hence $G_k$ is holomorphic and zero-free,
and $\abs{G_k}=\abs{\Ft_k}\le M$ on $\abs{z}=R$, so $\abs{G_k}\le M$ inside.
\end{proof}

Taking $Z_k$ in the open disk is what makes \eqref{eq:Bk} well defined. A zero
on $\abs{z}=R$ would make the corresponding factor of the form $0/0$ there.
Such a zero is simply not divided out, and it lies outside the disk on which
we work.

\begin{lemma}[Jensen]\label{lem:jensen}
Let $\mu=\ord_0\Ft_k$ and let $b^{(k)}_\mu$ be the first nonvanishing Taylor
coefficient. Set
\begin{equation}\label{eq:Lambda}
  \Lambda_k:=\exp\Bigl[\frac1{2\pi}\int_0^{2\pi}
  \log\abs{\Ft_k(R\rme^{\rmi\theta})}\dth\Bigr].
\end{equation}
Then $\abs{G_k(0)}=\Lambda_k$ and
\begin{equation}\label{eq:jensen}
  \log\Lambda_k=\log\abs{b^{(k)}_\mu}+\mu\log R
  +\sum_{0<\abs{z_j}<R}\log\frac{R}{\abs{z_j}} .
\end{equation}
\end{lemma}

\begin{proof}
Jensen's formula \cite[Theorem 15.18]{Rudin} applied to $\Ft_k(z)/z^\mu$,
whose value at the origin is $b^{(k)}_\mu$, gives \eqref{eq:jensen} with the
sum extended over $0<\abs{z_j}\le R$. A zero on $\abs{z}=R$ contributes
$\log(R/R)=0$, so the two sums agree and no choice of $R$ is needed. For the
first assertion, $\log\abs{G_k}$ is harmonic on $\abs{z}<R$, so its mean over
$\abs{z}=\rho$ equals $\log\abs{G_k(0)}$ for every $\rho<R$, and letting
$\rho\uparrow R$ and using $\abs{B_k}=1$ on the boundary circle identifies the
limit with $\log\Lambda_k$.
\end{proof}

\begin{theorem}[Rank-$r$ rigidity]\label{thm:harnack-r}
Let $\norm{F_k}_k\le1$ with at most $r$ zeros in $\abs{z}<R$, and let
$0<\rho<R$. Then for $k\ge k_0$ and $\abs{z}\le\rho$,
\begin{equation}\label{eq:harnackr}
  \abs{B_k(z)}\,M\Bigl(\frac{\Lambda_k}{M}\Bigr)^{\frac{R+\rho}{R-\rho}}
  \ \le\ \abs{\Ft_k(z)}\ \le\
  M\Bigl(\frac{\Lambda_k}{M}\Bigr)^{\frac{R-\rho}{R+\rho}} .
\end{equation}
At $r=0$ this reduces to \cref{lem:harnack}, since then $B_k\equiv1$ and
$\Lambda_k=\abs{F_k(0)}$.
\end{theorem}

\begin{proof}
Apply \cref{lem:harnack} to $G_k$, which is legitimate by
\cref{lem:blaschke}, and use $\abs{G_k(0)}=\Lambda_k$ from
\cref{lem:jensen}. Multiplying through by $\abs{B_k(z)}$ and using
$\abs{B_k}\le1$ for the upper bound gives \eqref{eq:harnackr}.
\end{proof}

\begin{corollary}\label{cor:blaschke-lb}
With $Z_k$ as in \cref{lem:blaschke},
$\abs{B_k(z)}\ge\bigl(\dist(z,Z_k)/2R\bigr)^{r}$ for $\abs{z}\le\rho$.
\end{corollary}

\begin{proof}See \cref{app:B2}.\end{proof}

The controlling quantity $\Lambda_k$ is the geometric mean of $\abs{\Ft_k}$ on
the circle $\abs{z}=R$, and by \eqref{eq:jensen} it packages the first
nonvanishing coefficient together with the positions of the zeros. It also
detects degeneration, and it does so in a usable form. At fixed $R$ one has
$\Lambda_k\to0$ if and only if $\Ft_k\to0$ uniformly on compact subsets of
$\abs{z}<R$: the upper bound in \eqref{eq:harnackr} gives one direction, and
uniform convergence to zero on the circle forces the geometric mean
\eqref{eq:Lambda} to zero in the other.

The constant $M$ in \eqref{eq:M} degenerates as $R\to\infty$, so
\cref{thm:harnack-r} is a statement on fixed compact sets and says nothing at
the scale $R\sim\sqrt{2k}$ of the boundary of the rescaled disk. That is the
region where escaping mass accumulates, and it is treated by a different
method in \cref{sec:escape}. The division of labour is deliberate, with harmonicity controlling the
interior and a variational argument the edge.

\subsection{The escape criterion}

One more object is needed. Let $\mu_k$ denote the rescaled Husimi measure, the
pushforward of $Q_{F_k}\,\rmd\nu_k$ under $\zeta\mapsto\sqrt{2k}\,\zeta$, a
positive measure on $\C$ of total mass $\norm{F_k}_k^2$. On any compact set
its density converges to $\pi^{-1}\abs{\Finf}^2\rme^{-\abs{z}^2}$ whenever
$\Ft_k\to\Finf$ locally uniformly, by dominated convergence using
\eqref{eq:M}.

The rigidity above is quantitative but requires a bound on the number of zeros
in a disk whose radius is fixed in advance. For families of fixed finite rank
the criterion can be stated purely in terms of the state, and in a form that
does not use \cref{thm:harnack-r} at all.

\begin{theorem}[Escape criterion]\label{thm:escape}
Let $\norm{F_k}_k\le1$ with $r_k\equiv r<\infty$. The following are
equivalent.
\begin{enumerate}
\item[(i)] $\Ft_k\to0$ locally uniformly on $\C$;
\item[(ii)] $\psi^{(k)}_0,\dots,\psi^{(k)}_r\to0$, that is, the amplitudes on
  the lowest $r+1$ weight vectors all tend to zero;
\item[(iii)] $\mu_k(\{\abs{z}\le R\})\to0$ for every $R>0$;
\item[(iv)] $\mu_k(\{\abs{z}\le R\})\to0$ for some $R>0$.
\end{enumerate}
\end{theorem}

\begin{proof}
(i)$\Rightarrow$(ii). Local uniform convergence gives $b^{(k)}_m\to0$ for
each fixed $m$. The relation between coefficients and amplitudes in
\cref{prop:identity}, together with $\wk(m)\to1$ from
\cref{lem:weight}(iv), turns this into $\psi^{(k)}_m\to0$.

(ii)$\Rightarrow$(i). Pass to a subsequence along which $\Ft_k\to L$, which is
possible by \cref{lem:normal}. If $L\not\equiv0$ then \cref{thm:semicont}
gives $r(L)\le r<\infty$, so by \cref{thm:collapse} we have
$L=P\rme^{az^2+bz+c}$ with $\deg P\le r$. The order of vanishing of $L$ at the
origin equals the multiplicity of $0$ as a root of $P$ and is therefore at
most $r$, so the coefficients $b_0,\dots,b_r$ of $L$ cannot all vanish. This
contradicts (ii). Hence every limit point is zero, and normality upgrades this
to convergence of the whole family.

(i)$\Rightarrow$(iii)$\Rightarrow$(iv) is immediate. For (iv)$\Rightarrow$(i),
suppose (i) fails. Then some subsequence converges to $\Finf\not\equiv0$, and
by dominated convergence
$\mu_k(\{\abs{z}\le R\})\to\pi^{-1}\int_{\abs{z}\le R}
\abs{\Finf}^2\rme^{-\abs{z}^2}\dA$ for every $R$. This vanishes for some $R$
only if $\Finf$ vanishes on that disk and hence identically.
\end{proof}

The proof of (ii)$\Rightarrow$(i) is short because the possible limits have
already been classified. All it uses is that a limit of rank at most $r$
cannot vanish to order greater than $r$ at the origin.

The criterion does not reduce to the vacuum amplitude once $r\ge1$. The basis
state $\ket{k,k+1}$ has stellar function proportional to $\zeta$, so
$F_k(0)=\psi^{(k)}_0=0$ for every $k$, yet the family converges to the Fock
state $\ket1$ and does not escape. Testing only $\psi_0$ would report escape,
and testing the lowest $r+1$ amplitudes is the least one can do.

At $r=0$ the criterion reads simply that a zero-free family escapes if and
only if the vacuum fidelity $\abs{\Braket{k,k|\psi_k}}^2$ tends to zero.

\subsection{No multi-lobed pathology}

One might expect that mass could be lost by splitting into two lobes that
separate while a valley between them deepens. Harmonicity forbids this at rank
zero, and \cref{cor:blaschke-lb} shows that finitely many zeros do not create
the freedom either.

\begin{theorem}\label{thm:nolobes}
Let $\norm{F_k}_k\le1$ with at most $r$ zeros in $\abs{z}<R$, and suppose
$\liminf_k\Lambda_k>0$. Then $\abs{\Ft_k}$ has at most $r$ interior local
minima in $\abs{z}<R$, each of them at a zero, and for every $\delta>0$ and
every $\rho<R$ there is $c=c(\delta,\rho,R,r)>0$ with
\begin{equation}
  \abs{\Ft_k(z)}\ \ge\ c
  \qquad\text{whenever }\abs{z}\le\rho\ \text{and}\ \dist(z,Z_k)\ge\delta .
\end{equation}
\end{theorem}

\begin{proof}
Off the zero set $\log\abs{\Ft_k}$ is harmonic, so it satisfies the minimum
principle and has no interior local minimum there, so the only candidates are the
zeros, of which there are at most $r$. The lower bound is
\eqref{eq:harnackr} combined with \cref{cor:blaschke-lb} and the assumed lower
bound on $\Lambda_k$.
\end{proof}

A valley separating two lobes at bounded rescaled distance therefore requires
driving a zero into it, and in the limit that valley is a zero of $\Finf$, of
which \cref{thm:collapse} allows exactly $\deg P$. The picture is summarised by
the factorisation $\Ft_k=B_kG_k$: the Blaschke factor is explicit and carries
all the zeros, and the remaining factor obeys the whole of the rank-zero
rigidity.

\section{Escape of mass as a phase transition}
\label{sec:escape}

\Cref{prop:fatou} allows the Fock norm of the limit to be strictly smaller
than the norms along the family, and \cref{thm:escape} shows that the extreme
case, in which the limit vanishes altogether, is detected by finitely many
amplitudes. What lies between is the subject of this section. We first give a criterion
for no loss at all, then a family that loses exactly half of its mass, and we
show that the halving is not an artefact of the construction.

We work with the measure $\mu_k$ introduced before \cref{thm:escape}, of total
mass $\norm{F_k}_k^2$, and with the local convergence recorded there. The
question is whether the total mass converges as well.

\subsection{Tightness}

The radial part of $\mu_k$ has an exact description, which reduces tightness
to a statement about the amplitudes alone.

\begin{lemma}\label{lem:beta}
Let $\norm{F_k}_k=1$ with amplitudes $\psi^{(k)}_m$ and write
$p^{(k)}_m=\abs{\psi^{(k)}_m}^2$. Then the law of $\abs{\zeta}^2$ under
$Q_{F_k}\,\rmd\nu_k$ is the mixture
\begin{equation}\label{eq:betamix}
  \sum_{m\ge0}p^{(k)}_m\,\Betadist(m+1,\,2k-1).
\end{equation}
\end{lemma}

\begin{proof}
In the orthonormal basis $e^{(k)}_m$ the density of $Q_{F_k}$ against
$\rmd\nu_k$ is
\[
  (1-\abs{\zeta}^2)^{2k}\Bigl|\sum_m\psi_m\,e^{(k)}_m(\zeta)\Bigr|^2 ,
\]
and the angular integration annihilates the cross terms because
$e^{(k)}_m$ is a multiple of $\zeta^m$. What remains is a sum of the radial
densities of the basis states. Writing $x=\abs{\zeta}^2$, the basis state
$\ket{k,k+m}$ contributes a density proportional to $x^m(1-x)^{2k-2}$, which
is $\Betadist(m+1,2k-1)$.
\end{proof}

The mean and variance of a $\Betadist(a,b)$ law are $a/(a+b)$ and
$ab/[(a+b)^2(a+b+1)]$, so the components of \eqref{eq:betamix} have
\begin{equation}\label{eq:betamoments}
  \mathbb{E}X=\frac{m+1}{m+2k},
  \qquad
  \Var X=\frac{(m+1)(2k-1)}{(m+2k)^2\,(m+2k+1)} ,
\end{equation}
and since $\abs{z}^2=2kX$ under the rescaling, the first of these says that
the level $m$ contributes mass at rescaled radius of order $\sqrt{m}$. Mass
therefore escapes exactly when the amplitudes migrate to high levels, which is
the content of the following.

\begin{theorem}\label{thm:tight}
The family $\{\mu_k\}$ is tight if and only if the coefficient distributions
$\{p^{(k)}\}$ are uniformly tight, that is, if and only if
\begin{equation}\label{eq:ui}
  \lim_{N\to\infty}\ \sup_k\ \sum_{m\ge N}p^{(k)}_m\ =\ 0 .
\end{equation}
\end{theorem}

\begin{proof}
Suppose \eqref{eq:ui} holds and let $\varepsilon>0$. By Markov's inequality
and \eqref{eq:betamoments},
\[
  \mathbb{P}\Bigl[X>\frac{R^2}{2k}\Bigr]
  \le\frac{2k}{R^2}\cdot\frac{m+1}{m+2k}\le\frac{m+1}{R^2},
\]
so splitting the mixture at a level $N$ gives
$\mu_k(\{\abs{z}>R\})\le N/R^2+\sum_{m\ge N}p^{(k)}_m$. Choosing $N$ by
\eqref{eq:ui} and then $R$ large makes this smaller than $\varepsilon$
uniformly in $k$.

Conversely, suppose \eqref{eq:ui} fails, so there are $\delta>0$ and
arbitrarily large $N$ with $\sum_{m\ge N}p^{(k)}_m>\delta$ for some $k\ge N$.
For $m\ge N$, Chebyshev's inequality and \eqref{eq:betamoments} give
\[
  \mathbb{P}\Bigl[X<\tfrac12\mathbb{E}X\Bigr]
  \le\frac{4\Var X}{(\mathbb{E}X)^2}
  =\frac{4(2k-1)}{(m+1)(m+2k+1)}\le\frac{4}{N},
\]
the last step because $m\ge N$ forces $(m+1)(m+2k+1)\ge N(2k-1)$. On the
complementary event $\abs{z}^2=2kX\ge k(m+1)/(m+2k)\ge N/3$ when
$N\le k$. Hence $\mu_k(\{\abs{z}^2\ge N/3\})\ge\delta(1-4/N)$, and letting
$N\to\infty$ contradicts tightness.
\end{proof}

\begin{corollary}\label{cor:tightsuff}
If $\sup_k\avg{\hat n}_k<\infty$, where $\hat n=K_0-k$ counts levels above
the lowest weight, then $\{\mu_k\}$ is tight and
$\norm{\Finf}_\Fock=\lim_k\norm{F_k}_k$. The condition is sufficient and not
necessary: for integer $k\ge3$, taking $p^{(k)}_0=1-1/\log k$ and
$p^{(k)}_k=1/\log k$ gives a tight family, since the mass at level $k$ tends to zero, while
$\avg{\hat n}_k=k/\log k$ is unbounded.
\end{corollary}

\begin{proof}
A uniform bound $\sum_m m\,p^{(k)}_m\le C$ gives
$\sum_{m\ge N}p^{(k)}_m\le C/N$ by Markov's inequality, uniformly in $k$, which
is \eqref{eq:ui}; tightness then turns the local convergence of $\mu_k$ into
convergence of the total mass.
\end{proof}

\Cref{cor:tightsuff} is easy to check and covers most families that arise in
practice, since $\avg{\hat n}$ is the mean level above the lowest weight
vector and is bounded for any family built from a fixed finite set of levels
or from coherent states with bounded labels.

\subsection{Partial escape}

Tightness fails in a way that is more interesting than a simple dichotomy
would allow. The next family keeps a nonzero limit and loses exactly half of
its mass, with the lost half leaving every compact set.

\begin{proposition}\label{prop:partial}
For $k>1/2$ set
\begin{equation}\label{eq:haloprop}
  F_k(\zeta)=2^{-1/2}\exp\bigl(\beta_k\zeta^3\bigr),
  \qquad \beta_k=\epsilon_k\,(2k)^{3/2},
\end{equation}
where $\epsilon_k>0$ is defined by the requirement $\norm{F_k}_k=1$. Such an
$\epsilon_k$ exists and is unique, and $\epsilon_k\to0$. Moreover
$\Ft_k\to2^{-1/2}$ locally uniformly on $\C$, while
\begin{equation}\label{eq:half}
  \mu_k\bigl(\{\abs{z}\le R\}\bigr)\ \longrightarrow\
  \tfrac12\bigl(1-\rme^{-R^2}\bigr)
  \qquad\text{for every }R>0 ,
\end{equation}
so that $\lim_{R\to\infty}\lim_k\mu_k(\{\abs{z}\le R\})=\tfrac12$ while
$\norm{F_k}_k=1$ for every $k$.
\end{proposition}

\begin{proof}
By \eqref{eq:key} the norm condition reads $1+S_k(\epsilon)=2$, where
\begin{equation}\label{eq:Ssum}
  S_k(\epsilon)=\sum_{j\ge1}\epsilon^{2j}\,\frac{(3j)!}{(j!)^2}\,\wk(3j) .
\end{equation}
The series converges for every $\epsilon>0$ and defines a strictly increasing
entire function of $\epsilon^2$ with $S_k(0)=0$ and
$S_k(\epsilon)\to\infty$, so $S_k(\epsilon_k)=1$ has a unique positive
solution. That $\epsilon_k\to0$ follows from a lower bound on a single term of
\eqref{eq:Ssum} at $j\approx\sqrt{k}$, using
\cref{lem:wquant}, and the computation is in \cref{app:semiclassical}.

Since $\Ft_k(z)=2^{-1/2}\exp(\epsilon_kz^3)$ and $\epsilon_k\to0$, the
convergence $\Ft_k\to2^{-1/2}$ is locally uniform. Hence
$\norm{\Finf}^2_\Fock=1/2$, and by the local convergence of $\mu_k$ recorded
above,
$\mu_k(\{\abs{z}\le R\})\to\pi^{-1}\int_{\abs{z}\le R}
\tfrac12\rme^{-\abs{z}^2}\dA=\tfrac12(1-\rme^{-R^2})$, which is
\eqref{eq:half} and increases to $1/2$ as $R$ grows. Since $\norm{F_k}_k=1$
for every $k$, half of the mass is not accounted for by any compact set.
\end{proof}

At the level of measures there is therefore no dichotomy. Full retention
occurs under \cref{cor:tightsuff}, full escape under \cref{thm:escape}, and
\cref{prop:partial} realises an intermediate case, the deficit being
$1-\norm{\Finf}^2_\Fock$ whenever $\norm{F_k}_k=1$.

The behaviour differs sharply from \cref{thm:escape}. At the level of
functions the behaviour is rigid, since either $\Ft_k\to0$ everywhere or the limit
is a nonzero Fock state, with nothing between. At the level of measures the
behaviour is not rigid, and the proportion of mass retained can take
intermediate values. The rigidity of \cref{sec:rigidity} governs the profile
of $\abs{\Ft_k}$ on compact sets, and says nothing about the region
$\abs{z}\sim\sqrt{2k}$ where the lost mass accumulates.

One feature of \cref{prop:partial} is left unexplained by the proof. The
constant $\epsilon_k$ is defined implicitly, through the intermediate value
theorem, yet numerically the product $\epsilon_k\sqrt{2k}$ converges to a
definite number, close to $0.9784$. That number has an origin, and it also
explains why the exponent three in \eqref{eq:haloprop} cannot be replaced by
one or two. Both points are the subject of the next subsection.

\subsection{The semiclassical mechanism}

The proofs so far have been soft, and the constant $0.9784$ left over from
\cref{prop:partial} does not come out of them. It comes instead from a
variational picture that is exact in the limit and that we now describe. The
results of this subsection are of a different character from the rest of the
paper. \Cref{prop:doublewell} and \cref{prop:saddlenode} are elementary
computations, and the statement they support, \cref{thm:phases}, is proved in
\cref{app:D5}.

Write the halo family of \eqref{eq:haloprop} in the form
$F_k(\zeta)\propto\exp(2k\lambda\zeta^3)$, so that $\lambda=\epsilon_k\sqrt{2k}$
is the coupling held fixed as $k$ grows. Then
$\abs{F_k(\zeta)}^2\propto\exp[4k\lambda\Ree(\zeta^3)]$, and multiplying by
the weight $(1-\abs{\zeta}^2)^{2k}$ gives the Husimi density against the
invariant measure as
\begin{equation}\label{eq:huspot}
  Q_{F_k}(\zeta)\ \propto\
  \exp\bigl[\,2k\,\Phi_\lambda(\zeta)\,\bigr],
  \qquad
  \Phi_\lambda(\zeta)=2\lambda\,\Ree(\zeta^3)+\log(1-\abs{\zeta}^2),
\end{equation}
so $2k$ plays the role of an inverse temperature and $\Phi_\lambda$ is a
potential on the disk. As $k\to\infty$ the mass concentrates on the maxima of
$\Phi_\lambda$, and the question of escape becomes a question about which
maximum dominates. The origin is always a critical point with
$\Phi_\lambda(0)=0$, and it is the one that the contraction resolves, since a
neighbourhood of the origin of radius $O(1/\sqrt{2k})$ is what the rescaling
magnifies. Any other maximum lies at a fixed distance from the origin and is
therefore invisible in the rescaled picture, and its mass escapes.

\begin{proposition}\label{prop:doublewell}
Let $\varphi(\varrho)=2\lambda \varrho^3+\log(1-\varrho^2)$ be the restriction of $\Phi_\lambda$
to the ray on which $\Ree(\zeta^3)$ is maximal. The conditions
$\varphi(\varrho)=\varphi'(\varrho)=0$, which express that an outer maximum is degenerate
with the one at the origin, are equivalent to
\begin{equation}\label{eq:degen}
  \frac{2\varrho^2}{3(1-\varrho^2)}=-\log(1-\varrho^2),
  \qquad
  \lambda=\frac{1}{3\varrho(1-\varrho^2)} ,
\end{equation}
and have the unique solution
\begin{equation}\label{eq:lamstar}
  \varrho^*=0.730471925\ldots,
  \qquad
  \abs{\zeta^*}^2=0.533589234\ldots,
  \qquad
  \lambda^*=0.978377935\ldots
\end{equation}
\end{proposition}

\begin{proof}
$\varphi'(\varrho)=6\lambda \varrho^2-2\varrho/(1-\varrho^2)$ vanishes for $\varrho>0$ exactly when
$\lambda=[3\varrho(1-\varrho^2)]^{-1}$. Substituting into $\varphi(\varrho)=0$ gives the first
equation of \eqref{eq:degen}. Its left side increases from $0$ to $\infty$ on
$(0,1)$ strictly faster than the right side does at first and slower later,
and a direct check gives a single crossing at $\varrho^*$. The numerical values are
computed to thirty digits in \cref{app:semiclassical}.
\end{proof}

\begin{theorem}[Phases of the halo family]\label{thm:phases}
Fix $\lambda>0$ and take $\epsilon=\lambda/\sqrt{2k}$ in \eqref{eq:haloprop}.
Then, with $S_k$ as in \eqref{eq:Ssum},
\begin{equation}\label{eq:phases}
  S_k(\epsilon)\longrightarrow0 \ \text{ if }\ \lambda<\lambda^*,
  \qquad
  S_k(\epsilon)\longrightarrow\infty \ \text{ if }\ \lambda>\lambda^* ,
\end{equation}
with $\lambda^*$ the constant of \cref{prop:doublewell}. Consequently all the
mass is retained below $\lambda^*$ and all of it escapes above, and the
constant of \cref{prop:partial} satisfies $\epsilon_k\sqrt{2k}\to\lambda^*$.
\end{theorem}

The proof is in \cref{app:D5} and is a Laplace analysis of the coefficient sum
\eqref{eq:Ssum} rather than of the phase-space integral. It produces its own
critical coupling, from $\log(1+3x)=2x$ together with
$\lambda=\rme^{3x}/\sqrt{27x}$, and the substitution
\begin{equation}\label{eq:xtor}
  x=\frac{\varrho^2}{3(1-\varrho^2)},
  \qquad\text{equivalently}\qquad
  1+3x=\frac{1}{1-\varrho^2},
\end{equation}
carries those two conditions into \eqref{eq:degen} exactly. The same
substitution carries the saddle-node of the coefficient problem, the double
root of $27\lambda^2x=(1+3x)^3$, into $\lambda_c=\sqrt3/2$ of
\cref{prop:saddlenode}. The two computations are one variational problem
written in two coordinates, one on phase space and one on the coefficient
index, and $x$ and $\varrho$ are related through the Beta mixture of
\cref{lem:beta}.

At $\lambda=\lambda^*$ exactly the behaviour is neither of the two in
\eqref{eq:phases}. Numerically $S_k\to3.272$, so a definite fraction
$1/(1+3.272)=0.234$ of the mass is retained; establishing that limit needs the
next order of the Laplace expansion, which we do not carry out. This is not
the situation of \cref{prop:partial}, where $\epsilon_k$ is tuned so that
$S_k=1$ for every $k$. There $\lambda_k=\epsilon_k\sqrt{2k}$ approaches
$\lambda^*$ from below and the retained fraction is exactly one half, and the
same tuning produces any value in $(0,1)$.

The reading of \cref{thm:phases} that matters here is that escape of mass is
not a failure of the estimates. It is a transition in a variational problem
that the contraction cannot see, because the second maximum of $\Phi_\lambda$
sits at a fixed distance from the origin while the rescaling magnifies a
neighbourhood shrinking like $1/\sqrt{2k}$.

\subsection{Why the exponent must exceed two}

The halo in \eqref{eq:haloprop} uses $\zeta^3$, and the exponent is not
incidental. The same potential with a general power has an outer critical
point only when the power exceeds two, and this is the same cutoff that
appears in \cref{thm:collapse} through the degree of $g$.

\begin{proposition}\label{prop:saddlenode}
Let $\varphi_n(\varrho)=2\lambda \varrho^n+\log(1-\varrho^2)$ with $n\ge1$. The saddle-node
conditions $\varphi_n'(\varrho)=\varphi_n''(\varrho)=0$, at which an outer critical point
is born, are satisfied only for
\begin{equation}\label{eq:rc}
  \varrho_c^2=\frac{n-2}{n},
  \qquad
  \lambda_c=\frac{1}{n\,\varrho_c^{\,n-2}(1-\varrho_c^2)} .
\end{equation}
The condition has no solution for $n=1$ and degenerates to $\varrho_c=0$ for $n=2$.
For $n=3$ it gives $\varrho_c=1/\sqrt3$ and $\lambda_c=\sqrt3/2=0.866025\ldots$,
and there $\varphi_3(\varrho_c)=\tfrac13+\log\tfrac23<0$, so the outer critical
point is born below the value at the origin and is at first subdominant.
\end{proposition}

\begin{proof}
From $\varphi_n'=2n\lambda \varrho^{n-1}-2\varrho/(1-\varrho^2)$ the first condition gives
$\lambda=[nr^{n-2}(1-\varrho^2)]^{-1}$. Substituting into
$\varphi_n''=2n(n-1)\lambda \varrho^{n-2}-2(1+\varrho^2)/(1-\varrho^2)^2$ and simplifying yields
$(n-1)(1-\varrho^2)=1+\varrho^2$, that is $\varrho_c^2=(n-2)/n$. The remaining assertions are
substitutions.
\end{proof}

An outer well therefore exists only for $n\ge3$, and these are exactly the
exponents for which no rescaling of the coefficient makes $\rme^{a_nz^n}$
survive the contraction, since the scaling window of \cref{thm:collapse}
admits only $n\le2$.

The two cutoffs are the same cutoff. A mode with $n\le2$ has nowhere to put its mass
except the neighbourhood of the origin that the contraction magnifies, and it
survives as a displacement or a squeezing. A mode with $n\ge3$ has a second
well at a fixed distance, and as its coupling grows the mass moves there and
is lost. The cutoff at degree two in \eqref{eq:limitform}, which arises in
\cref{thm:collapse} from the Borel--Carath\'eodory inequality and the
quadratic Fock weight, has this dynamical counterpart.

\Cref{fig:potential} shows $\varphi_3$ at the four values of $\lambda$ that
separate the regimes.

\begin{figure}[t]
\centering
\begin{tikzpicture}[x=8.6cm,y=4.6cm]
  \draw[->,gray!70] (0,0) -- (1.00,0) node[right,black,font=\small] {$\varrho$};
  \draw[->,gray!70] (0,-0.72) -- (0,0.24) node[above,black,font=\small]
    {$\varphi_3(\varrho)$};
  \foreach \x/\lab in {0.2/0.2,0.4/0.4,0.6/0.6,0.8/0.8}
    {\draw[gray!70] (\x,-0.010) -- (\x,0.010);}
  \foreach \x/\lab in {0.4/0.4,0.6/0.6,0.8/0.8}
    {\node[font=\scriptsize,below,inner sep=2.5pt] at (\x,-0.012) {$\lab$};}
  \foreach \y/\lab in {-0.6/{-0.6},-0.4/{-0.4},-0.2/{-0.2},0.1/{0.1}}
    {\draw[gray!70] (-0.006,\y) -- (0.006,\y);
     \node[font=\scriptsize,left] at (-0.008,\y) {$\lab$};}
  \draw[gray!45,dotted] (0.57735,-0.72) -- (0.57735,0.20);
  \draw[gray!45,dotted] (0.73047,-0.72) -- (0.73047,0.20);
  \node[font=\scriptsize,gray!75,anchor=south] at (0.57735,0.20) {$\varrho_c$};
  \node[font=\scriptsize,gray!75,anchor=south] at (0.73047,0.20) {$\varrho^*$};
  \draw[thick,domain=0:0.882,samples=200,smooth]
    plot(\x,{2*0.60*\x*\x*\x+ln(1-\x*\x)});
  \draw[thick,black!55,domain=0:0.882,samples=200,smooth]
    plot(\x,{2*0.866025*\x*\x*\x+ln(1-\x*\x)});
  \draw[very thick,densely dashed,domain=0:0.882,samples=200,smooth]
    plot(\x,{2*0.978378*\x*\x*\x+ln(1-\x*\x)});
  \draw[thick,dotted,domain=0:0.882,samples=200,smooth]
    plot(\x,{2*1.12*\x*\x*\x+ln(1-\x*\x)});
  \node[font=\scriptsize,anchor=west] at (0.895,-0.639) {$\lambda=0.60$};
  \node[font=\scriptsize,anchor=west] at (0.895,-0.289) {$\lambda_c=0.866$};
  \node[font=\scriptsize,anchor=west] at (0.895,-0.141) {$\lambda^*=0.978$};
  \node[font=\scriptsize,anchor=west] at (0.895, 0.046) {$\lambda=1.12$};
\end{tikzpicture}
\caption{The radial potential $\varphi_3(\varrho)=2\lambda \varrho^3+\log(1-\varrho^2)$ of
\eqref{eq:huspot} at four couplings. Below $\lambda_c=\sqrt3/2$ the origin is
the only critical point. At $\lambda_c$ an outer critical point is born, by
\cref{prop:saddlenode}, but with $\varphi_3<0$, so it is subdominant. At
$\lambda^*=0.9784$ the outer maximum reaches the value at the origin, which is
the degeneracy of \cref{prop:doublewell}. Above $\lambda^*$ it dominates and
the mass leaves the region that the contraction magnifies.}
\label{fig:potential}
\end{figure}
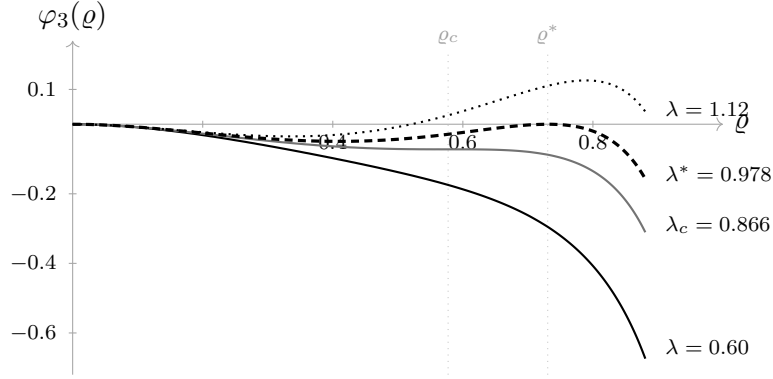

\section{The limit map}
\label{sec:limitmap}

\Cref{thm:collapse} describes the limits, and \cref{sec:rigidity} says when
there is one at all. What remains is how much of a family determines its
limit, and which limits are attained. Write
\begin{equation}\label{eq:Lr}
  \mathcal L_r:=\bigl\{P(z)\rme^{az^2+bz+c}:\ \deg P\le r,\ \abs{a}<\tfrac12\bigr\}
\end{equation}
for the possible nonzero limits of rank at most $r$. This is a family with
$r+3$ complex parameters: the $r+1$ coefficients of $P$ with the overall scale
absorbed into $c$, together with $a$ and $b$. One might guess that $r+3$
Taylor coefficients determine an element of $\mathcal L_r$, and the guess is
wrong.

\subsection{How much of the state survives}

\begin{lemma}[Logarithmic derivative]\label{lem:logder}
Let $L=P\rme^{Q}\in\mathcal L_r$ with $P(0)\ne0$. Then
\begin{equation}\label{eq:logder}
  \frac{L'}{L}=\frac{P'+PQ'}{P}=:\frac{N}{P},
  \qquad \deg N\le r+1,\ \ \deg P\le r,
\end{equation}
a rational function of type $(r+1,r)$, and $L$ is determined by $L'/L$
together with the value $L(0)$.
\end{lemma}

\begin{proof}
The identity is immediate, and $\deg(PQ')\le r+1$ since $\deg Q'\le1$. Given
$L'/L$ and $L(0)$, the function $L$ is recovered by integrating the
logarithmic derivative on a neighbourhood of the origin, where $L$ does not
vanish.
\end{proof}

We use one classical fact about rational functions. If $N_1/D_1$ and $N_2/D_2$
are of type $(m,n)$, have denominators nonvanishing at the origin, and agree
in their first $m+n+1$ Taylor coefficients, then $N_1D_2-N_2D_1$ is a
polynomial of degree at most $m+n$ whose first $m+n+1$ Taylor coefficients all
vanish, hence identically zero, so the two rational functions are equal.

\begin{theorem}\label{thm:injective}
Two elements of $\mathcal L_r\cup\{0\}$ whose first $2r+3$ Taylor coefficients
agree are equal.
\end{theorem}

\begin{proof}
Let $L_1,L_2\in\mathcal L_r$ agree in $b_0,\dots,b_{2r+2}$. The case where one
of them is zero is covered because an element of $\mathcal L_r$ vanishes at the
origin to order at most $r<2r+3$. They therefore have the same order of
vanishing $\nu\le r$ at the origin, and dividing both by $z^{\nu}$ replaces
$\mathcal L_r$ by $\mathcal L_{r-\nu}$ and leaves at least
$2r+3-\nu\ge2(r-\nu)+3$ matching coefficients, so we may assume
$L_i(0)\ne0$.

By \cref{lem:logder} each $L_i'/L_i$ is rational of type $(r+1,r)$ with
denominator nonvanishing at the origin. The first $2r+2$ Taylor coefficients of
$L'/L$ are determined by the first $2r+3$ of $L$, so the two logarithmic
derivatives agree in $(r+1)+r+1=2r+2$ coefficients and are therefore equal, and by
\cref{lem:logder} together with $L_1(0)=L_2(0)=b_0$ so are $L_1$ and $L_2$.
\end{proof}

\begin{theorem}\label{thm:unique}
Let $\norm{F_k}_k\le1$ with $r_k\equiv r<\infty$. Then $\Ft_k$ converges
locally uniformly on $\C$ if and only if the $2r+3$ numbers
$b^{(k)}_0,\dots,b^{(k)}_{2r+2}$ converge.
\end{theorem}

\begin{proof}
Necessity is the Cauchy integral formula. For sufficiency, every limit point
lies in $\mathcal L_r\cup\{0\}$ by \cref{thm:collapse,thm:semicont}, and by
hypothesis all limit points share their first $2r+3$ coefficients, so
\cref{thm:injective} makes them equal. A normal family with a single limit
point converges.
\end{proof}

In physical terms, the contraction limit retains exactly the lowest $2r+3$
weight amplitudes and forgets the rest. At $r=0$ these are the vacuum, the
one-quantum and the two-quantum amplitudes, which is consistent with the limit
being a Gaussian state, determined by a displacement and a squeezing.

The count cannot be reduced, and the obstruction is explicit.

\begin{theorem}[Sharpness]\label{thm:sharp}
For every $r\ge0$ there exist $L_1\ne L_2$ in $\mathcal L_r$ whose first
$2r+2$ Taylor coefficients agree. For even $r=2s$ one may take
\begin{equation}\label{eq:evensharp}
  L_{1,2}(z)=p_s(\pm\lambda z^2)\,\rme^{\pm\lambda z^2/2},
  \qquad
  p_s(w)=\sum_{j=0}^{s}\frac{(2s-j)!}{j!\,(s-j)!}(-w)^j ,
\end{equation}
where $p_s$ is the denominator of the diagonal Pad\'e approximant to
$\rme^{w}$ and $0<\lambda<1$. For odd $r$ the construction reduces to a root of
a Hankel determinant. The proofs and the explicit constants are in
\cref{app:Nr}.
\end{theorem}

Together, \cref{thm:injective} and \cref{thm:sharp} identify $2r+3$ as the
exact number of Taylor coefficients that determine an element of
$\mathcal L_r$. The smallest case is worth writing out. At $r=1$ the two functions
\begin{equation}\label{eq:r1witness}
  L_1(z)=(z-1)\rme^{az^2},
  \qquad
  L_2(z)=\omega^2(z-\omega)\,
  \rme^{(a-\frac34-\frac{\sqrt3}{4}\rmi)z^2
        +(-\frac32+\frac{\sqrt3}{2}\rmi)z},
  \qquad \omega=\rme^{-2\pi\rmi/3},
\end{equation}
agree in their coefficients of $z^0,\dots,z^3$ and differ at $z^4$, and there
are values of $a$ for which both exponents satisfy the constraint
$\abs{a}<\tfrac12$ of \eqref{eq:Lr}, so the pair lies in $\mathcal L_1$ itself,
and not just in the larger family obtained by dropping the type restriction.

Convergence of the whole family is a genuine restriction, and its failure is
not pathological.

\begin{proposition}\label{prop:arc}
Let $F_k$ be the normalised Perelomov coherent state with
$\zeta_0(k)=\alpha(k)/\sqrt{2k}$ and $\alpha(k)=2+\sin(\log k)$. Then
$\norm{F_k}_k=1$ for every $k$, the family depends real analytically on $k$,
and its limit set is the arc
$\bigl\{\rme^{\bar\alpha z-\abs{\alpha}^2/2}:\alpha\in[1,3]\bigr\}$.
\end{proposition}

\begin{proof}
Each $F_k$ contracts to the Glauber state with label $\alpha(k)$, and
$\alpha(k)$ ranges over $[1,3]$ and returns to each value for arbitrarily
large $k$, since $\log k$ is unbounded and $\sin$ is periodic.
\end{proof}

By \cref{thm:unique} the failure of convergence in \cref{prop:arc} is exactly
the failure of the three numbers $b_0,b_1,b_2$ to converge, and those three
are free to oscillate independently. Nonconvergence is a consequence of the
freedom, and not a sign of irregularity in the family.

The arc in \cref{prop:arc} is not an accident of the example. Let
$\norm{F_k}_k\le1$ with $r_k\equiv r<\infty$. The limit set is closed in the
topology of local uniform convergence and, by \cref{lem:normal}, contained in a
compact set, so it is compact; each of its elements lies in
$\mathcal L_r\cup\{0\}$ by \cref{thm:collapse} applied to the corresponding
subsequence. It is connected whenever the family depends continuously on $k$,
since the limit set at infinity of a continuous curve in a compact metric
space is a continuum, so a limit set consisting of exactly two points is
impossible. We note that $\mathcal L_r\cup\{0\}$ is not itself closed, as
$\rme^{a_nz^2}$ with $\abs{a_n}\uparrow\tfrac12$ shows, which is why the
argument runs through the limit set rather than through $\mathcal L_r$.

\subsection{Surjectivity}

Every admissible limit occurs, and a canonical family realises it. The
construction is the inverse of the rescaling, and the norm constraint one
might expect to have to impose turns out to hold automatically.

\begin{theorem}[Canonical section]\label{thm:section}
For $G\in\Fock$ define $\sigma_k(G)(\zeta):=G(\sqrt{2k}\,\zeta)$. Then
$\sigma_k(G)\in\Ak$ and $\Tk\sigma_k(G)=G$ identically in $k$, with
\begin{equation}\label{eq:sectionnorm}
  \norm{\sigma_k(G)}_k^2=\sum_m\abs{b_m}^2m!\,\wk(m)
  \ \uparrow\ \norm{G}^2_\Fock .
\end{equation}
In particular $\norm{G}_\Fock\le1$ implies $\norm{\sigma_k(G)}_k\le1$ for
every $k$, so the norm constraint is automatic. The range of the limit map on
the unit ball of $\Ak$ is therefore the \emph{closed} unit ball of
$\mathcal L_r\cup\{0\}$, and $\sigma_k$ depends holomorphically on $a$, $b$,
$c$ and the coefficients of $P$.
\end{theorem}

\begin{proof}
The identity $\Tk\sigma_k=\mathrm{id}$ is immediate from the definitions, and
\eqref{eq:sectionnorm} is \cref{prop:identity} applied to $\sigma_k(G)$, whose
rescaled Taylor coefficients are those of $G$. The inequality and the monotone
convergence are \cref{lem:weight}(i) and (iv). Holomorphy in the parameters is
clear since $\sigma_k$ acts by substitution. In the language of
\eqref{eq:diag}, $\sigma_k$ is diagonal with singular values $\wk(m)^{1/2}\le1$
in the two orthonormal bases, so it is a contraction that tends to the
identity.
\end{proof}

If one insists on $\norm{F_k}_k=1$ exactly rather than $\le1$, the interior of
the ball is reached as well by adding a cubic halo as in \cref{prop:partial}
and choosing its coefficient so that the norm is exact. We establish existence
only, since the dependence of that coefficient on $(a,b,c)$ and on $P$ is not shown
to be continuous, so this version supplies no continuous section. For a
continuous section, use \cref{thm:section} with $\norm{\cdot}_k\le1$.

\section{Discussion}
\label{sec:discussion}

The stellar rank is a complete invariant on the plane in the sense that its
vanishing characterises the Gaussian states, and it is not complete on the
disk. The results above locate the difference precisely, and the location is
not where the symmetry sits but where the weight does.

On the plane the chain runs as follows. The reproducing kernel gives the bound
$\abs{F(z)}\le\norm{F}\rme^{\abs{z}^2/2}$. The exponent two in
$\rme^{-\abs{z}^2}$ is what makes the Borel--Carath\'eodory inequality return
$\deg g\le2$ in \cref{thm:collapse}, so the zero-free class is
$\rme^{az^2+bz+c}$, a family with two free complex parameters, and the group
acting simply transitively on it is generated by displacements and
squeezings. Each step in that chain is a consequence of the exponent two.

At finite $k$ that number does not exist. The disk is a bounded domain, the
weight $(1-\abs{\zeta}^2)^{2k-2}$ imposes no growth restriction of the kind
Borel--Carath\'eodory needs. The general relation between the weight sequence
of a Bargmann-type space and the growth of its elements is developed in
\cite{VPDS2012}. In our case the zero-free class in $\Ak$ is
infinite-dimensional, because the measure is finite, so $H^{\infty}(\D)\subset\Ak$ and
$\rme^{g}$ is a zero-free element for every bounded holomorphic $g$. The
quadratic weight, and with it the number two, is produced by the contraction,
through $(1-\abs{z}^2/2k)^{2k}\to\rme^{-\abs{z}^2}$ in
\eqref{eq:weightkernel}.

\Cref{sec:escape} adds the dynamical version of the same statement. A mode
$\rme^{a_nz^n}$ survives the contraction exactly when $n\le2$, and the
semiclassical potential of \eqref{eq:huspot} acquires a second well exactly
when $n\ge3$. The two cutoffs coincide because they are the same cutoff.

This answers the objection recorded in \cref{sec:intro}. Klimov and coauthors
argue that hyperbolic phase space needs no criteria imported from flat or
spherical geometry \cite{Klimov2026}, and we agree. What the present results add is a reason. The completeness of the flat
criterion is a property of the flat weight and is manufactured by the limit,
so at finite $k$ there is nothing to export. The
anomaly they report, that the zero set leaves a zero-free factor
undetermined, is what one should expect on any bounded domain.

A zero set is not a unitary invariant, and the discrete series is realised
analytically in more than one way. The Perelomov realisation on the disk used
here and the Barut--Girardello realisation on the plane carry the same
representation and are related by
\begin{equation}\label{eq:BVMlaplace}
  G(1/\rho;k)=\frac{\rho^{2k}}{\sqrt{\Gamma(2k)}}
  \int_0^\infty z^{2k-1}F(z;k)\,\rme^{-\rho z}\,\rmd z ,
  \qquad \Ree\rho>0,\ \abs{\rho}>1,
\end{equation}
a Laplace transform, but a unitary map between Hilbert spaces of
holomorphic functions places no constraint on where those functions vanish.
The transform is due to Brif, Vourdas and Mann, who obtain \eqref{eq:BVMlaplace}
and record its inverse \cite[equations (23) and (24)]{BVM1996}; the review
\cite[section 11.3]{Vourdas2006} restates it. Their normalisation of the
Barut--Girardello monomials is $n!\,\Gamma(n+2k)$, which is $\Gamma(2k)$ times
the $h_n$ used here, and \eqref{eq:BVMlaplace} is an isometry in theirs; with
$h_n=n!\,(2k)_n$ the prefactor becomes $\rho^{2k}/\Gamma(2k)$.

Comparing the two realisations directly is instructive. The Barut--Girardello
space has monomial norms
$h_n=n!\,(2k)_n$ and reproducing kernel ${}_0F_1(2k;\bar z'z)$, of order one
and type one \cite[section 11.1]{Vourdas2006}. Hadamard's theorem in the form
used in \cref{thm:collapse}, with the exponent one in place of two, gives the
zero-free class $\{\rme^{bz+c}\}$ subject to the type constraint
$\abs{b}<1$. Since amplitudes and Taylor coefficients are related there by
$\psi_n=c_n\sqrt{h_n}$, which is \eqref{eq:psic} with the monomial norms of
that space in place of those of $\Ak$, and
$\rme^{bz}$ has $c_n=b^n/n!$, its amplitudes are
\begin{equation}\label{eq:BGtoPer}
  \psi_n=\frac{b^n}{n!}\sqrt{h_n}=b^n\sqrt{\frac{(2k)_n}{n!}} ,
\end{equation}
which are exactly the amplitudes of the Perelomov coherent state
$\ket{k,\zeta_0}$ with $\zeta_0=b$. Normalisation gives the binomial series
$\sum_n\abs{b}^{2n}(2k)_n/n!=(1-\abs{b}^2)^{-2k}$, finite precisely for
$\abs{b}<1$, so the type constraint reproduces the label disk. In the
Barut--Girardello realisation the zero-free class is therefore exactly the
$SU(1,1)$ orbit of the lowest weight vector, whereas in the Perelomov
realisation it is infinite-dimensional and contains that orbit along with much
else, and the two descriptions are of the same states.

Two qualifications are needed here. The identification of the disk of admissible $b$ with
the disk of Perelomov labels holds in the standard normalisation, the one in
which $K_+$ acts as multiplication by $z$, and a rescaling $z\mapsto\alpha z$
changes the radius and leaves the orbit itself untouched, so the invariant
statement is the one about the orbit. And the reproducing kernel of the
Barut--Girardello space, being $ {}_0F_1(2k;x)=\Gamma(2k)x^{(1-2k)/2}
I_{2k-1}(2\sqrt x)$, has infinitely many zeros on the negative real axis, so
Barut--Girardello coherent states have infinite stellar rank in their own
realisation. There is no contradiction, since the zero-free class there is the
Perelomov orbit and not the Barut--Girardello one, but the fact is
counterintuitive enough to state.

None of the ingredients is new. That the zero set of a disk
representation leaves an outer factor undetermined appears in
\cite{VBM1996} for $k=1/2$, in the language of inner and outer functions, and
again in \cite[section 10.4]{Vourdas2006}. The order and type of the
Barut--Girardello space are recorded in \cite[section 11.1]{Vourdas2006}. What
seems not to have been said is that combining the two identifies the zero-free
class in that realisation with a group orbit.

Several questions are left open. \Cref{thm:rank-eq} assumes $r_k\equiv
r<\infty$, and when $r_k\to\infty$ the counting measure of
\cref{thm:measures} still converges but no integer invariant survives, so it
is unclear what replaces the scaling condition. A zero at scale
$k^{-\gamma}$ with $\gamma<1/2$ simply escapes, as \cref{sec:rank} shows,
and a finer description of the rate at which the associated mass leaves a
fixed compact set is not attempted here. In \cref{thm:phases} the behaviour
exactly at $\lambda^*$ is left open, since it requires the next order of the
Laplace expansion of \eqref{eq:Ssum2} rather than the exponential rate alone.

Further afield, the regime $r\to\infty$ with $r/k$ fixed is a thermodynamic
limit of the stellar
hierarchy, and on the Pad\'e side of \cref{app:Nr} the natural tool there is
the Riemann--Hilbert analysis of Hermite--Pad\'e approximation to $\rme^{z}$;
whether the two ends meet is not known to us. And since Klimov and coauthors
prove that Perelomov coherent states have strictly positive Wigner functions
\cite{Klimov2026}, the identification above gives, in the Barut--Girardello
realisation, that rank zero implies Wigner positivity. Whether the converse
holds, that is, whether Wigner positivity forces a Perelomov coherent state,
would be an $SU(1,1)$ analogue of Hudson's theorem, and we do not know
whether one exists.

The scope adopted at the outset leaves out mixed states, several modes, and
the continuous and principal series. The Majorana representation for mixed
states of \cite{SEBraun2020} suggests where the first of these might begin.

\section*{Data availability statement}

No new data were created or analysed in this study. The scripts that
reproduce every constant, table and figure coordinate reported here are
openly available at \texttt{doi:10.5281/zenodo.22313935}, as described in
\cref{app:E4}.

\appendix
\numberwithin{equation}{section}
\numberwithin{table}{section}
\numberwithin{figure}{section}
\renewcommand{\thetable}{\thesection\arabic{table}}
\renewcommand{\thefigure}{\thesection\arabic{figure}}

\section{The weight sequence}
\label{app:weight}

Every statement about $\wk$ that we use follows from a single observation,
that the ratio $(2k)^m/(2k)_m$ is a product of $m$ elementary factors. We
collect the consequences here because they are invoked repeatedly, in
different combinations, and at points of the argument that are far apart.

\subsection{Proof of \texorpdfstring{\Cref{lem:weight}}{the weight lemma}}
\label{app:A1}

Recall from \eqref{eq:key} that
\begin{equation}\label{eq:wproduct}
  \wk(m)\;=\;\frac{(2k)^m}{(2k)_m}
  \;=\;\prod_{j=0}^{m-1}\Bigl(1+\frac{j}{2k}\Bigr)^{-1},
\end{equation}
the second form following from $(2k)_m=\prod_{j=0}^{m-1}(2k+j)$, with the
convention that an empty product equals one. The factors are indexed by $j$
and depend on $k$ only through the combination $j/2k$. All four assertions of
\cref{lem:weight} therefore reduce to elementary statements about a single
factor, and we treat them in that order.

\paragraph{(i) Range and the first two values.}
Each factor $(1+j/2k)^{-1}$ lies in $(0,1]$ for $j\ge0$ and $k>0$, and equals
one exactly when $j=0$. Multiplying gives $0<\wk(m)\le1$. For $m=0$ the
product is empty and for $m=1$ it consists of the single factor $j=0$, so
\begin{equation}\label{eq:w01}
  \wk(0)=\wk(1)=1 \qquad\text{for every }k>0 .
\end{equation}
For $m\ge2$ the factor $j=1$ is present and is strictly less than one, so
$\wk(m)<1$ strictly. Property (i) is what makes the section $\sigma_k$ of \cref{thm:section}
a contraction, and \eqref{eq:w01} is what makes \cref{prop:exact2} exact
rather than asymptotic. Reading (i) as a comparison of the two norms needs
care, since $\Tk F$ need not lie in $\Fock$, and we return to this in
\cref{app:A2}.

\paragraph{(ii) Monotonicity in $m$.}
Consecutive values differ by one factor,
\begin{equation}\label{eq:wratio-m}
  \frac{\wk(m+1)}{\wk(m)}=\Bigl(1+\frac{m}{2k}\Bigr)^{-1},
\end{equation}
which equals one for $m=0$ and is strictly less than one for $m\ge1$. Hence
$\wk$ is nonincreasing on $m\ge0$ and strictly decreasing on $m\ge1$, with the
plateau $\wk(0)=\wk(1)$ at the bottom. The decay is superexponential in $m$
at fixed $k$, as the quantitative form below makes precise, and this is why a
state whose amplitudes are spread over high weight vectors is penalised in
$\Ak$ relative to $\Fock$.

\paragraph{(iii) Monotonicity in $k$.}
Fix $m\ge2$ and regard \eqref{eq:wproduct} as a function of $k$. Each factor
$(1+j/2k)^{-1}$ is strictly increasing in $k$ when $j\ge1$ and constant when
$j=0$. Since $m\ge2$ guarantees that at least one factor with $j\ge1$ occurs,
the product is strictly increasing in $k$. For $m\le1$ the value is constant
by \eqref{eq:w01}.

 The consequence for the estimates in the body is described
after \cref{lem:weight}.

\paragraph{(iv) The limit.}
For fixed $m$ the product in \eqref{eq:wproduct} has finitely many factors,
each tending to one as $k\to\infty$, so $\wk(m)\to1$. By (iii) the
convergence is monotone from below. Together with \cref{prop:identity} this
says that the contraction turns the $\Ak$ norm into the Fock norm one
coefficient at a time, monotonically, and that the two norms agree exactly on
the lowest two levels at every finite $k$. \qed

\medskip

The three properties (i), (iii) and (iv) are qualitative, and for the
semiclassical estimates of \cref{app:semiclassical} we need a quantitative
form of the decay. It follows from the same product by taking logarithms.

\begin{lemma}\label{lem:wquant}
For all $k>0$ and $m\ge0$,
\begin{equation}\label{eq:wquant}
  \exp\Bigl(-\frac{m(m-1)}{4k}\Bigr)
  \;\le\;\wk(m)\;\le\;
  \exp\Bigl(-\frac{m(m-1)}{4k}+\frac{m^3}{24k^2}\Bigr).
\end{equation}
\end{lemma}

\begin{proof}
Taking logarithms in \eqref{eq:wproduct},
\[
  -\log\wk(m)=\sum_{j=0}^{m-1}\log\Bigl(1+\frac{j}{2k}\Bigr).
\]
For $x\ge0$ one has $x-\tfrac12x^2\le\log(1+x)\le x$. Applying this with
$x=j/2k$ and summing, and using
$\sum_{j<m}j=\tfrac12m(m-1)$ and $\sum_{j<m}j^2\le\tfrac13m^3$, gives
\[
  \frac{m(m-1)}{4k}-\frac{m^3}{24k^2}
  \;\le\;-\log\wk(m)\;\le\;\frac{m(m-1)}{4k},
\]
which is \eqref{eq:wquant} after exponentiating.
\end{proof}

The lower bound in \eqref{eq:wquant} is the one we use. It shows that the
penalty $\wk(m)$ imposes on the level $m$ is Gaussian in $m$ with variance of
order $k$, so that levels up to $m\sim\sqrt{k}$ are essentially unpenalised
while levels well beyond that are suppressed. This is the analytic content
behind the observation, in \cref{app:semiclassical}, that the coefficient
series associated with a cubic halo has its second maximum near
$m\sim k$ rather than at small $m$.

\subsection{Derivation of \texorpdfstring{\Cref{prop:identity}}{the key
identity}, and the contraction as a diagonal map}
\label{app:A2}

The identity \eqref{eq:key} is a change of variables, carried out here in a
form that keeps track of which object lives in which space, because
the three sets of coefficients used in the body, namely $c_m$, $b_m$ and
$\psi_m$, are easy to confuse.

Let $F\in\Ak$ have Taylor coefficients $c_m$. By \cref{def:contraction} the
rescaled function $\Tk F$ has coefficients
\begin{equation}\label{eq:bfromc}
  b_m=(2k)^{-m/2}c_m .
\end{equation}
Substituting into \eqref{eq:Aknorm} and using
$\Gamma(2k)/\Gamma(m+2k)=1/(2k)_m$ together with the definition
$\wk(m)=(2k)^m/(2k)_m$,
\begin{equation}
  \norm{F}_k^2=\sum_m\abs{c_m}^2\frac{m!}{(2k)_m}
  =\sum_m\abs{b_m}^2(2k)^m\frac{m!}{(2k)_m}
  =\sum_m\abs{b_m}^2\,m!\,\wk(m),
\end{equation}
which is \eqref{eq:key}. The amplitudes follow from
$\psi_m=c_m[m!/(2k)_m]^{1/2}$ and \eqref{eq:bfromc}:
\begin{equation}\label{eq:psib}
  \psi_m=b_m\sqrt{m!\,\wk(m)} .
\end{equation}
We use \eqref{eq:psib} in both directions. In \cref{thm:escape} the criterion
is stated in terms of the amplitudes, which are what an experiment sees, while
the proof runs on the Taylor coefficients. Since $\wk(m)\to1$ at fixed $m$ by
\cref{lem:weight}(iv), the two are interchangeable in any statement about a
fixed finite set of levels.

\paragraph{The contraction in the two orthonormal bases.}
The reading of \eqref{eq:key} as a deformation of one quadratic form into
another, given after \cref{prop:identity}, has a shorter form in bases.
The functions
$e^{(k)}_m(\zeta)=\zeta^m[(2k)_m/m!]^{1/2}$ are orthonormal in $\Ak$ and
$f_m(z)=z^m/\sqrt{m!}$ are orthonormal in $\Fock$, and a direct computation
from \eqref{eq:bfromc} gives
\begin{equation}\label{eq:diag}
  \Tk\,e^{(k)}_m=\wk(m)^{-1/2}f_m,
  \qquad
  \sigma_k f_m=\wk(m)^{1/2}\,e^{(k)}_m ,
\end{equation}
with $\sigma_k$ the map of \cref{thm:section}. Thus $\Tk$ is diagonal in these
bases with singular values $\wk(m)^{-1/2}\ge1$, so it expands, and its formal
inverse $\sigma_k$ is diagonal with singular values $\wk(m)^{1/2}\le1$, so it
contracts.

Three facts used in the body can be read off \eqref{eq:diag} directly. First, $\Tk$ is exact on the lowest two levels for every finite $k$,
by $\wk(0)=\wk(1)=1$, which is \cref{prop:exact2}. Second, $\sigma_k$ maps the
unit ball of $\Fock$ into the unit ball of $\Ak$ without any further
condition, which is the assertion in \cref{thm:section} that the norm
constraint is automatic. Third, the singular values of $\sigma_k$ increase to
one, so $\norm{\sigma_kG}_k\uparrow\norm{G}_\Fock$, and the range of the limit
map is the closed rather than the open unit ball.

\subsection{Exactness at the bottom of the ladder, and a table}
\label{app:A3}

\Cref{prop:exact2} asserts that the contraction is exactly correct on the two
lowest levels at every finite $k$. With \eqref{eq:diag} in hand this is immediate, and we record it in the
sharper form that the two lowest levels are the only ones on which it holds.

\begin{proposition}\label{prop:exactmax}
Let $V\subseteq\Ak$ be a closed subspace on which $\Tk$ is isometric onto its
image in $\Fock$. Then
$V\subseteq\operatorname{span}\{e^{(k)}_0,e^{(k)}_1\}$, and on that
two-dimensional subspace $\Tk$ is isometric for every $k>1/2$.
\end{proposition}

\begin{proof}
By \eqref{eq:diag}, a vector $v=\sum_m v_m e^{(k)}_m$ has
$\norm{\Tk v}^2_\Fock-\norm{v}^2_k=\sum_m\abs{v_m}^2[\wk(m)^{-1}-1]$. Every
summand is nonnegative, and by \cref{lem:weight}(i) the bracket vanishes
exactly for $m\le1$ and is strictly positive for $m\ge2$. The difference
therefore vanishes if and only if $v_m=0$ for all $m\ge2$. Isometry on the
remaining two-dimensional subspace is \eqref{eq:w01}.
\end{proof}

The concrete consequence, the rank-one state whose zero sits at
$\alpha/\sqrt{2k}$, is spelled out after \cref{prop:exact2} in the body and is
not repeated here.

Exactness fails first at level two, and the failure is of order $1/k$. From
\eqref{eq:wproduct},
\begin{equation}\label{eq:w2exact}
  \wk(2)=\frac{2k}{2k+1},\qquad\text{so}\qquad 1-\wk(2)=\frac{1}{2k+1},
\end{equation}
in agreement with the expansion $1-\wk(m)=m(m-1)/4k+O(k^{-2})$ implied by
\cref{lem:wquant}.

\Cref{tab:weights} lists $\wk(m)$ over the range of $m$ and $k$ that occurs in
the examples, and \cref{lem:weight} can be read off it. Along a row the values
increase towards one, which is monotonicity in $k$. Down a column they fall
away, at first slowly and then very fast, which is the Gaussian
suppression of \cref{lem:wquant}. The transition between the two regimes sits
near $m\sim\sqrt{k}$, as that lemma predicts, and this is the reason the
coefficient sums of \cref{app:semiclassical} have a second maximum far out
rather than a monotone tail.

\begin{table}[t]
\centering
\begin{tabular}{rlllll}
\toprule
& \multicolumn{5}{c}{$k$}\\
\cmidrule(l){2-6}
$m$ & $1$ & $5$ & $50$ & $500$ & $5000$\\
\midrule
$1$ & $1$ & $1$ & $1$ & $1$ & $1$ \\
$2$ & $0.666667$ & $0.909091$ & $0.990099$ & $0.999001$ & $0.999900$ \\
$3$ & $0.333333$ & $0.757576$ & $0.970685$ & $0.997007$ & $0.999700$ \\
$5$ & $0.044444$ & $0.416250$ & $0.906166$ & $0.990065$ & $0.999001$ \\
$10$ & $2.565\times10^{-5}$ & $0.029831$ & $0.646366$ & $0.956133$ & $0.995512$ \\
$30$ & $1.306\times10^{-25}$ & $1.779\times10^{-11}$ & $0.018761$ & $0.649999$ & $0.957473$ \\
\bottomrule
\end{tabular}
\caption{The weight sequence $\wk(m)=(2k)^m/(2k)_m$. The first row is
identically one by \eqref{eq:w01}. Values increase in $k$ along each row by \cref{lem:weight}(iii) and decrease
in $m$ down each column by \cref{lem:weight}(ii).}
\label{tab:weights}
\end{table}

\section{Covariance, Jensen, and the convergence of zeros}
\label{app:jensen}

\subsection{The cocycle identity and the reduction to the centred rescaling}
\label{app:B1}

\begin{proof}[Proof of \cref{lem:cocycle}]
Write $w=g\cdot\zeta=(a\zeta+b)/(\bar b\zeta+\bar a)$. Then
\[
  1-\abs{w}^2
  =\frac{\abs{\bar b\zeta+\bar a}^2-\abs{a\zeta+b}^2}
        {\abs{\bar b\zeta+\bar a}^2},
\]
and expanding the numerator,
\[
  \bigl(\abs{b}^2\abs{\zeta}^2+\bar b a\zeta+\bar ab\bar\zeta+\abs{a}^2\bigr)
  -\bigl(\abs{a}^2\abs{\zeta}^2+a\bar b\zeta+\bar ab\bar\zeta+\abs{b}^2\bigr)
  =\bigl(\abs{b}^2-\abs{a}^2\bigr)\abs{\zeta}^2+\abs{a}^2-\abs{b}^2 ,
\]
the cross terms cancelling in pairs. Using $\abs{a}^2-\abs{b}^2=1$ this equals
$1-\abs{\zeta}^2$, which is \eqref{eq:cocycle}.
\end{proof}

\begin{proof}[Proof of the covariance assertions of \cref{lem:cocycle}]
By \eqref{eq:cocycle},
\[
  (1-\abs{\zeta}^2)^{k}\abs{(U_gF)(\zeta)}
  =\Bigl(\frac{1-\abs{\zeta}^2}{\abs{\bar b\zeta+\bar a}^{2}}\Bigr)^{k}
   \abs{F(g\cdot\zeta)}
  =(1-\abs{g\cdot\zeta}^2)^{k}\abs{F(g\cdot\zeta)},
\]
which is the asserted pointwise invariance of the Husimi amplitude.
Integrating the square of this identity against the invariant measure
$\rmd\nu_k$ gives $\norm{U_gF}_k=\norm{F}_k$, and $U_g$ is invertible with
inverse $U_{g^{-1}}$, hence unitary. Finally $U_gF$ vanishes at $\zeta$
exactly when $F$ vanishes at $g\cdot\zeta$, and $g$ is a bijection of $\D$, so
the number of zeros in $\D$ is unchanged.
\end{proof}

\Cref{prop:scale} follows: replacing $F_k$ by $U_{g_k}^{-1}F_k$ changes
neither the norm, nor the zero count, nor the Husimi amplitude, so a family
and its image satisfy the hypotheses and the conclusions of every statement in
the paper simultaneously. For the centred rescaling, taking logarithms and using
$\log(1-x)=-x+O(x^2)$ as $\lambda_k\to0$,
\[
  2k\log\bigl(1-\lambda_k^2\abs{u}^2\bigr)
  =-2k\lambda_k^2\abs{u}^2+O\bigl(k\lambda_k^4\bigr),
\]
and the error term is $O\bigl((2k\lambda_k^2)\lambda_k^2\bigr)$, which
vanishes whenever $2k\lambda_k^2$ stays bounded. The limit of
$(1-\lambda_k^2\abs{u}^2)^{2k}$ is therefore
$\exp(-\lim_k2k\lambda_k^2\abs{u}^2)$, a nondegenerate Gaussian exactly when
$2k\lambda_k^2$ has a limit in $(0,\infty)$, identically one when the limit is
$0$, and vanishing off the origin when it is $\infty$.

\subsection{Jensen's formula and the Blaschke factor}
\label{app:B2}

The definition \eqref{eq:Lambda} of $\Lambda_k$ involves the circle
$\abs{z}=R$, and one might expect to need $R$ chosen so that $\Ft_k$ has no
zero there, which cannot be arranged uniformly in $k$ because the zeros move.
No such choice is needed, because the version of Jensen's formula in
\cite[Theorem 15.18]{Rudin} counts the zeros of the \emph{closed} disk
$\abs{z}\le R$, and a zero on the boundary circle contributes the factor
$R/\abs{z_j}=1$. The integral in \eqref{eq:Lambda} converges in any case,
since $\log\abs{\Ft_k}$ has at worst logarithmic singularities on the circle.

As a check, take $f(z)=(z-a)(z-R)$ with $0<\abs{a}<R$, which has one zero
inside and one exactly on the circle. Jensen predicts
$\log\abs{f(0)}+\log(R/\abs{a})=\log(\abs{a}R)+\log(R/\abs{a})=2\log R$, with
no contribution from the boundary zero. At $R=2$ and
$a=0.7\,\rme^{0.4\rmi}$ the circle mean of $\log\abs{f}$ evaluates to
$1.38629436111989$, against $2\log2=1.38629436111989$.

What does require the open disk is the Blaschke factor itself.

\begin{proof}[Proof of \cref{cor:blaschke-lb}]
Let $n=\#Z_k\le r$. For $\abs{z}\le\rho<R$ and $z_j\in Z_k$,
\[
  \abs{R^2-\bar z_jz}\ \le\ R^2+\abs{z_j}\abs{z}\ \le\ 2R^2 ,
\]
so each factor of \eqref{eq:Bk} satisfies
\[
  \Bigl|\frac{R(z-z_j)}{R^2-\bar z_jz}\Bigr|
  \ \ge\ \frac{R\abs{z-z_j}}{2R^2}
  \ =\ \frac{\abs{z-z_j}}{2R}
  \ \ge\ \frac{\dist(z,Z_k)}{2R} .
\]
Since $\dist(z,Z_k)\le2R$ the quantity on the right lies in $[0,1]$, so
raising it to the power $n$ and then to the larger power $r$ only decreases
it:
\[
  \abs{B_k(z)}\ \ge\ \Bigl(\frac{\dist(z,Z_k)}{2R}\Bigr)^{n}
  \ \ge\ \Bigl(\frac{\dist(z,Z_k)}{2R}\Bigr)^{r} . \qedhere
\]
\end{proof}

The restriction of $Z_k$ to the open disk in \cref{lem:blaschke} is what
makes $\abs{B_k}\equiv1$ on $\abs{z}=R$ hold without exception. If a zero lay
on the circle, its factor would be $0/0$ there. Numerically, for $R=2$ and
zeros at $1.3\,\rme^{2.1\rmi}$ and $0.2$, the modulus of $B_k$ on the circle
deviates from one by at most $2\times10^{-31}$ at working precision, and
$\abs{B_k(0.5)}=0.1136<1$ inside, as the maximum principle requires.

\subsection{Convergence of the zero counting measures}
\label{app:B3}

We prove \cref{lem:zeroconv} and \cref{thm:measures}. Both are consequences of
Hurwitz's theorem, applied first near each zero of the limit and then on the
complement.

\begin{proof}[Proof of \cref{lem:zeroconv}]
Let $w$ be a zero of $\Finf$ of multiplicity $\nu$ and let $\varepsilon>0$ be
small enough that $\Finf$ has no other zero in $\abs{z-w}\le\varepsilon$. Then
$\Finf\ne0$ on the circle $\abs{z-w}=\varepsilon$, and by compactness
$\inf\abs{\Finf}>0$ there. Uniform convergence on that circle gives
$\abs{\Ft_k-\Finf}<\abs{\Finf}$ on it for large $k$, so Rouch\'e's theorem
applies and $\Ft_k$ has exactly $\nu$ zeros in the disk. For the second
assertion, if $K$ is compact with $\Finf\ne0$ on $K$, cover $K$ by finitely
many disks on whose boundaries the same argument applies.
\end{proof}

\begin{proof}[Proof of \cref{thm:measures}]
Let $\varphi\in C_c(\C)$ with support in $\abs{z}\le\rho$, and let
$w_1,\dots,w_p$ be the zeros of $\Finf$ in $\abs{z}\le\rho+1$, finitely many
because $\Finf\not\equiv0$ is holomorphic. Fix $\eta>0$. Choose
$\varepsilon\in(0,\eta)$ small enough that the disks
$D_i=\{\abs{z-w_i}<\varepsilon\}$ are disjoint, contain no other zero of
$\Finf$, and satisfy
$\abs{\varphi(z)-\varphi(w_i)}<\eta$ on $D_i$, which is possible by uniform
continuity.

By \cref{lem:zeroconv}, for large $k$ each $D_i$ contains exactly $\nu_i$
zeros of $\Ft_k$, where $\nu_i$ is the multiplicity of $w_i$, and the compact
set $\{\abs{z}\le\rho+\tfrac12\}\setminus\bigcup_iD_i$ contains none. Hence
\[
  \Bigl|\int\varphi\,\rmd\Nz_k-\int\varphi\,\rmd\Nz_{\Finf}\Bigr|
  =\Bigl|\sum_i\ \sum_{z_j\in D_i}\bigl(\varphi(z_j)-\varphi(w_i)\bigr)\Bigr|
  \ \le\ \eta\sum_i\nu_i ,
\]
the two measures having no other mass in the support of $\varphi$. Since
$\sum_i\nu_i$ is a fixed finite number and $\eta$ was arbitrary, the two
integrals agree in the limit.
\end{proof}

The argument uses only that $\Finf\not\equiv0$, which makes its zero set
locally finite. An alternative route is through the Poincar\'e--Lelong
formula, which identifies $\Nz_k$ with
$(2\pi)^{-1}\Delta\log\abs{\Ft_k}$ and reduces the claim to
$L^1_{\mathrm{loc}}$ convergence of $\log\abs{\Ft_k}$, a standard property of
subharmonic functions that are locally uniformly bounded above and do not
tend to $-\infty$. That route generalises to situations where the zeros are
not locally finite, and is the natural one if the counting measures are to be
handled by potential theory, but it is not needed here.

\section{Sharpness of the number \texorpdfstring{$N(r)=2r+3$}{N(r)=2r+3}}
\label{app:Nr}

The upper bound of \cref{thm:injective} is proved in the body. This appendix
supplies the constructions behind \cref{thm:sharp}, which show that $2r+3$
cannot be replaced by $2r+2$. The even and the odd cases look different, and
we treat them separately before showing that they rest on the same object.

\subsection{A construction from a vanishing determinant}
\label{app:C1}

The witnesses of \cref{thm:sharp} all arise from one construction, which we
state first. Fix $R(z)=\alpha z^2+\beta z$ and write $\tau_n=[z^n]\rme^{R(z)}$,
so that $\tau_n$ is a polynomial in $\alpha$ and $\beta$.

\begin{lemma}\label{lem:construction}
Let $r\ge1$ and suppose $\alpha$ satisfies $\abs{\alpha}<\tfrac12$ and
\begin{equation}\label{eq:Dr}
  D_r(\alpha,\beta):=\det\bigl(\tau_{1+p+i}\bigr)_{0\le p,i\le r}=0
\end{equation}
for some $\beta$, with $(\alpha,\beta)\ne(0,0)$. Then there exist
$L_1\ne L_2$ in $\mathcal L_r$ whose Taylor coefficients agree in the
positions $0,\dots,2r+1$.
\end{lemma}

\begin{proof}
Write $P=\sum_{i=0}^{r}c_iz^i$ and ask that the coefficients of
$P\rme^{R}$ vanish in the positions $r+1,\dots,2r+1$, that is
\begin{equation}\label{eq:Csystem}
  \sum_{i=0}^{r}c_i\,\tau_{m-i}=0,
  \qquad m=r+1,\dots,2r+1 .
\end{equation}
This is a homogeneous system of $r+1$ equations in the $r+1$ unknowns $c_i$,
and reversing the order of the columns, $i\mapsto r-i$, turns its matrix into
the Hankel matrix of \eqref{eq:Dr}. By hypothesis its determinant vanishes, so
a nonzero solution $P$ exists.

Set $L_2:=P\rme^{R}$ and let $L_1$ be the truncation of $L_2$ to degree $r$,
a polynomial. The two agree in the positions $0,\dots,r$ because $L_1$ was
defined that way; they agree in the positions $r+1,\dots,2r+1$ because $L_2$
vanishes there by \eqref{eq:Csystem} while $L_1$ has degree at most $r$. Hence
$L_1-L_2=O(z^{2r+2})$.

Both belong to $\mathcal L_r$: $L_1$ is a polynomial of degree at most $r$,
which is the case $a=b=0$ of \eqref{eq:Lr}, and $L_2=P\rme^{R}$ has
$\deg P\le r$ with $\abs{\alpha}<\tfrac12$. They are distinct because $R\ne0$
and $P\ne0$, so $L_2$ is not a polynomial. By \cref{thm:injective} the order
of $L_1-L_2$ is exactly $2r+2$, since a higher order would force
$L_1=L_2$.
\end{proof}

Everything therefore reduces to producing a root of $D_r$ inside
$\abs{\alpha}<\tfrac12$. For odd $r$ we do this in \cref{app:C3}. For even $r$
a second and more explicit construction is available, which we give first
because it also supplies the constant that settles \eqref{eq:Cshankel}.

\subsection{Even rank: the diagonal Pad\'e denominator}
\label{app:C2}

Both existence arguments below use the Hankel determinants
\begin{equation}\label{eq:Cshankel}
  \det C_s,
  \qquad
  C_s:=\Bigl(\frac{1}{(a+b+1)!}\Bigr)_{0\le a,b\le s} ,
\end{equation}
of the coefficient sequence of $\rme^{w}$, and both need them to be nonzero.
\Cref{cor:detCs} below establishes that, and also evaluates them, as a
consequence of the error constant of the diagonal Pad\'e approximant.

Let $p_s$ be as in \eqref{eq:evensharp}. Up to the constant $s!/(2s)!$ it is
the denominator $Q_{s,s}$ of the diagonal Pad\'e approximant to $\rme^{w}$,
for which see \cite{BGM1996}. The defining property is that
$p_s(w)\rme^{w}-p_s(-w)$ vanishes to high order, and for the diagonal case the
order and the leading constant are explicit.

\begin{lemma}\label{lem:Cs}
For every $s\ge0$,
\begin{equation}\label{eq:Cs}
  p_s(w)\,\rme^{w}-p_s(-w)=E_s\,w^{2s+1}+O(w^{2s+2}),
  \qquad
  E_s=(-1)^{s}\,\frac{s!}{(2s+1)!}\ \ne\ 0 .
\end{equation}
\end{lemma}

\begin{proof}
Hermite's integral representation of the Pad\'e error for the exponential
\cite[Theorem 3.11]{HairerWanner2002} reads, in the diagonal case,
\begin{equation}\label{eq:hermiteint}
  Q_{s,s}(w)\rme^{w}-P_{s,s}(w)
  =\frac{(-1)^{s}}{(2s)!}\,w^{2s+1}\int_0^1
   \rme^{(1-t)w}\,t^{s}(1-t)^{s}\,\rmd t ,
\end{equation}
with $Q_{s,s}(w)=\tfrac{s!}{(2s)!}p_s(w)$ and $P_{s,s}(w)=\tfrac{s!}{(2s)!}p_s(-w)$,
the second identity being the symmetry that follows from
$\rme^{w}\rme^{-w}=1$. Setting $w=0$ inside the integral gives the Euler
beta integral $\int_0^1t^s(1-t)^s\rmd t=(s!)^2/(2s+1)!$, so the coefficient of
$w^{2s+1}$ on the left is $(-1)^s(s!)^2/[(2s)!(2s+1)!]$. Dividing by
$s!/(2s)!$ gives
\begin{equation}\label{eq:Csclosed}
  E_s=(-1)^{s}\,\frac{s!}{(2s+1)!}\ \ne\ 0 .
\end{equation}
\end{proof}

The first values are $E_1=-\tfrac16$, $E_2=\tfrac1{60}$, $E_3=-\tfrac1{840}$,
$E_4=\tfrac1{15120}$ and $E_5=-\tfrac1{332640}$, in agreement with
\eqref{eq:Csclosed} and with symbolic computation for $s\le6$. Throughout
this appendix $E_s$ denotes this scalar and $C_s$ the matrix
\eqref{eq:Cshankel}.

\begin{proof}[Proof of \cref{thm:sharp}, even $r=2s$]
Put $w=\lambda z^2/2$ and let $L_{1,2}$ be as in \eqref{eq:evensharp}, so that
$L_{1,2}(z)=p_s(\pm2w)\rme^{\pm w}$. By \eqref{eq:Cs} the difference
$L_1-L_2$ vanishes to order $2s+1$ in $w$ and no further, hence to order
$2(2s+1)=4s+2=2r+2$ in $z$ and no further. Both functions lie in
$\mathcal L_r$: the polynomial factors have degree $2s=r$ in $z$, the
exponents are $\pm\lambda z^2/2$, and the constraint $\abs{a}<\tfrac12$ of
\eqref{eq:Lr} holds for $\lambda<1$.
\end{proof}

\begin{corollary}\label{cor:detCs}
For every $s\ge0$,
\begin{equation}\label{eq:detCs}
  \abs{\det C_s}=\prod_{i=0}^{s}\frac{(i!)^2}{(2i)!\,(2i+1)!}\ \ne\ 0 .
\end{equation}
\end{corollary}

\begin{proof}
The classical relation between the leading Pad\'e error coefficient and the
Hankel determinants of the underlying sequence
\cite[equation (1.11)]{BGM1996} states, for the diagonal entry $[s/s]$ of
$\rme^{w}$, that the error constant of $Q_{s,s}\rme^{w}-P_{s,s}$ is
$\pm\det C_s/\det C_{s-1}$. By \eqref{eq:Csclosed} that constant is
$\tfrac{s!}{(2s)!}E_s$ in absolute value, that is
$(s!)^2/[(2s)!(2s+1)!]$, so
\[
  \Bigl|\frac{\det C_s}{\det C_{s-1}}\Bigr|
  =\frac{(s!)^2}{(2s)!\,(2s+1)!} .
\]
Since $\det C_0=1$ and every factor is nonzero, induction gives
\eqref{eq:detCs}. The sign alternates in blocks of two,
$\operatorname{sign}\det C_s=(-1)^{\lceil s/2\rceil}$, as symbolic
computation for $s\le6$ confirms.
\end{proof}

At $r=0$ one has $p_0\equiv1$ and the pair reduces to
$L_{1,2}(z)=\rme^{\pm\lambda z^2/2}$, whose difference is
$\lambda z^2+O(z^6)$, of order $2=2r+2$. For $r=2$ the pair is
$L_{1,2}(z)=(2\mp\lambda z^2)\rme^{\pm\lambda z^2/2}$,
whose difference is $-\tfrac43(\lambda z^2/2)^3+O(z^{10})$, of order six in
$z$, as $2r+2=6$ requires. This pair has a different shape from the one
produced by \cref{lem:construction}, since both of its members carry an
exponential factor.

\subsection{Odd rank: a root of the Hankel determinant}
\label{app:C3}

For odd $r$ the symmetric ansatz \eqref{eq:evensharp} is unavailable and we
use \cref{lem:construction}, so what has to be shown is that $D_r$ has a
suitable root.

\begin{lemma}\label{lem:Dr}
For odd $r$ the polynomial $\alpha\mapsto D_r(\alpha,1)$ has degree exactly
$(r+1)^2/2$, with leading coefficient $\pm(\det C_s)^2\ne0$ where
$s=(r-1)/2$. At $\beta=0$ and $\alpha=1$ one has
\begin{equation}\label{eq:Drblock}
  D_r(1,0)=\pm\bigl(\det C_s\bigr)^2,
  \qquad s=\frac{r-1}{2},
\end{equation}
with $C_s$ the matrix of \eqref{eq:Cshankel},
which is nonzero by \eqref{eq:Cshankel}.
\end{lemma}

\begin{proof}
At $\beta=0$ and $\alpha=1$ the coefficients are $\tau_n=1/(n/2)!$ for even $n$
and $\tau_n=0$ for odd $n$, so the entry $\tau_{1+p+i}$ is nonzero exactly when
$p+i$ is odd. For odd $r$ the index set $\{0,\dots,r\}$ contains $(r+1)/2$
even and $(r+1)/2$ odd values. Permuting rows and columns so that the even
indices come first turns the matrix into
$\begin{psmallmatrix}0&X\\ Y&0\end{psmallmatrix}$, where the entry of $X$ at
$(a,b)$ comes from $p=2a$, $i=2b+1$ and equals $\tau_{2a+2b+2}=1/(a+b+1)!$, and
the entry of $Y$ at $(a,b)$ comes from $p=2a+1$, $i=2b$ and is the same.
Hence $X=Y=C_s$ and the determinant is $\pm(\det C_s)^2$, the sign being that
of the permutation used above.

For the degree, give $\alpha$ the weight $2$ and $\beta$ the weight $1$. Then
$\tau_n=\sum_{j}\alpha^{j}\beta^{\,n-2j}/[j!\,(n-2j)!]$ is weighted homogeneous of
weight $n$, so every term of the determinant has weight
$\sum_{p}(1+p+\sigma(p))=(r+1)+2\cdot\tfrac12 r(r+1)=(r+1)^2$ and $D_r$ is
weighted homogeneous of that weight. Writing
$D_r(\alpha,\beta)=\sum_j d_j\alpha^{j}\beta^{(r+1)^2-2j}$, the degree of
$D_r(\cdot,1)$ is the largest $j$ with $d_j\ne0$. Since $(r+1)^2$ is even for
odd $r$, the top possible value $j=(r+1)^2/2$ occurs with $\beta^{0}$, so its
coefficient is read off at $\beta=0$, and the block computation just made
identifies it as $\pm(\det C_s)^2$, nonzero by \cref{cor:detCs}. Symbolic
computation for $r=1,3,5$ confirms the degrees $2$, $8$ and $18$.
\end{proof}

\begin{proof}[Proof of \cref{thm:sharp}, odd $r$]
By \cref{lem:Dr} the polynomial $D_r(\cdot,1)$ is not identically zero and has
degree $(r+1)^2/2\ge2$, so it has at least one root, and no root is $\alpha=0$
because $D_r(0,1)=\det C_r\ne0$ by \cref{cor:detCs}. If a root lies
outside $\abs{\alpha}<\tfrac12$, replace $z$ by $\lambda z$ with $\lambda$
small, which multiplies $\alpha$ by $\lambda^2$ and $\beta$ by $\lambda$ and
leaves \eqref{eq:Dr} homogeneous in the required sense. \Cref{lem:construction}
then supplies the pair.
\end{proof}

The smallest case is explicit. At $r=1$ one has
$D_1(\alpha,1)=-\bigl(\alpha^2+\tfrac1{12}\bigr)$, whose roots
$\alpha=\pm\rmi/(2\sqrt3)$ already satisfy $\abs{\alpha}=0.2887<\tfrac12$, and
\cref{lem:construction} returns
\begin{equation}\label{eq:r1pair}
  L_1(z)=1+\bigl(\tfrac12-\alpha\bigr)z,
  \qquad
  L_2(z)=\bigl(1-(\alpha+\tfrac12)z\bigr)\rme^{\alpha z^2+z},
\end{equation}
which agree in the coefficients of $z^0,\dots,z^3$ and differ at $z^4$. The
same computation at $r=3$ and $r=5$ gives roots of modulus $0.109$ and
$0.047$, again inside the disk, with the order of $L_1-L_2$ equal to $8$ and
$12$ as required.

The even and the odd constructions therefore rest on one and the same family
of determinants. \Cref{cor:detCs} makes this explicit: the error constant of
the diagonal Pad\'e approximant, which drives the even case, is exactly the
ratio of the consecutive Hankel determinants that drive the odd one.

\section{Semiclassical analysis of the escape transition}
\label{app:semiclassical}

\subsection{The halo family: existence of \texorpdfstring{$\epsilon_k$}{eps k}
and the estimate \texorpdfstring{$\epsilon_k\to0$}{eps k to 0}}
\label{app:D1}

Recall \eqref{eq:Ssum},
\begin{equation}\label{eq:Ssum2}
  S_k(\epsilon)=\sum_{j\ge1}u_j(\epsilon,k),
  \qquad
  u_j(\epsilon,k)=\epsilon^{2j}\,\frac{(3j)!}{(j!)^2}\,\wk(3j) ,
\end{equation}
so that the normalisation in \cref{prop:partial} reads $S_k(\epsilon_k)=1$.

\begin{lemma}\label{lem:Sk}
For every $k>1/2$ the series \eqref{eq:Ssum2} converges for all
$\epsilon\in\C$, and $\epsilon\mapsto S_k(\epsilon)$ is strictly increasing on
$(0,\infty)$ with $S_k(0)=0$ and $S_k(\epsilon)\to\infty$. Hence
$S_k(\epsilon_k)=1$ has a unique positive solution, and $\epsilon_k\to0$ as
$k\to\infty$.
\end{lemma}

\begin{proof}
For fixed $k$ the ratio of consecutive terms is
\[
  \frac{u_{j+1}}{u_j}
  =\epsilon^2\,\frac{(3j+1)(3j+2)(3j+3)}{(j+1)^2}\cdot
   \frac{(2k)^3}{(3j+2k)(3j+2k+1)(3j+2k+2)},
\]
which behaves as $\epsilon^2(2k)^3/j^2$ for large $j$ and therefore tends to
zero. The series is entire in $\epsilon$. All terms are positive and strictly
increasing in $\epsilon>0$, and $u_1=6\epsilon^2\wk(3)\to\infty$, so the
remaining assertions follow.

For the second assertion it suffices to show that $S_k(\epsilon)\to\infty$ for each fixed
$\epsilon>0$, since $S_k$ is increasing and $S_k(\epsilon_k)=1$. Bound
$S_k$ below by the single term at $j=\lfloor\sqrt{k}\rfloor$. By
\cref{lem:wquant},
$\wk(3j)\ge\exp[-3j(3j-1)/4k]\ge\exp(-9j^2/4k)$, and
$(3j)!/(j!)^2\ge j!$, since
\[
  \frac{(3j)!}{(j!)^2}=\frac{(3j)!}{(2j)!}\binom{2j}{j}
  \ \ge\ \frac{(3j)!}{(2j)!}=\prod_{i=1}^{j}(2j+i)\ \ge\ (2j)^{j}\ \ge\ j^{j}
  \ \ge\ j! .
\]
Hence
\[
  \log u_j\ \ge\ 2j\log\epsilon+\log j!-\frac{9j^2}{4k}
  \ \ge\ 2j\log\epsilon+j\log j-j-\frac{9j^2}{4k} .
\]
With $j=\lfloor\sqrt{k}\rfloor$ the last term is bounded and
\begin{equation}\label{eq:crossover}
  \log u_j\ \ge\ \sqrt{k}\Bigl(2\log\epsilon+\tfrac12\log k-1\Bigr)-\tfrac94 ,
\end{equation}
which tends to $+\infty$ for every fixed $\epsilon>0$.
\end{proof}

Estimate \eqref{eq:crossover} governs how the family behaves numerically, in
two respects. The bracket becomes positive only when
$\tfrac12\log k>1-2\log\epsilon$, that is when
$k\gtrsim\rme^{2}\epsilon^{-4}$, and this crossover is severe, occurring near
$k=1.2\times10^6$ at $\epsilon=0.05$ and near $k=2.6\times10^7$ at
$\epsilon=0.023$. The summand $u_j$ is also not monotone in $j$. At
\emph{fixed} $\epsilon$ the ratio $u_{j+1}/u_j\approx\epsilon^2(2k)^3/j^2$,
valid for $j\gg2k$, drops below one only after $j\approx\epsilon(2k)^{3/2}$.
At the critical scaling $\epsilon\asymp\lambda/\sqrt{2k}$ that approximation
does not apply, since the relevant $j$ is comparable to $2k$; there the
turning points are the roots of \eqref{eq:stationary}, the second sitting at
$j=x_+(\lambda)\,2k$, which is $0.3813\,(2k)$ at $\lambda=\lambda^*$. In both
regimes the series falls, then rises again, and the second rise carries most
of the mass. \Cref{app:numerics} returns to the practical consequence.

\subsection{The degenerate double well}
\label{app:D2}

The two conditions of \cref{prop:doublewell} are
$\varphi(\varrho)=\varphi'(\varrho)=0$ for $\varphi(\varrho)=2\lambda \varrho^3+\log(1-\varrho^2)$.
Eliminating $\lambda$ through $\varphi'(\varrho)=0$, that is
$\lambda=[3\varrho(1-\varrho^2)]^{-1}$, and substituting into $\varphi(\varrho)=0$ gives
\begin{equation}\label{eq:degen2}
  h(\varrho):=\frac{2\varrho^2}{3(1-\varrho^2)}+\log(1-\varrho^2)=0 .
\end{equation}
On $(0,1)$ one has $h(0^+)=0^-$ in the sense that
$h(\varrho)=-\tfrac13\varrho^2+O(\varrho^4)$ near the origin, while $h(\varrho)\to+\infty$ as
$\varrho\to1^-$ because the first term diverges as $(1-\varrho^2)^{-1}$ and the second only
logarithmically. Since
$h'(\varrho)=\tfrac{4\varrho}{3(1-\varrho^2)^2}-\tfrac{2\varrho}{1-\varrho^2}
=\tfrac{2\varrho}{3(1-\varrho^2)^2}\bigl[2-3(1-\varrho^2)\bigr]$,
$h$ decreases on $(0,1/\sqrt3)$ and increases on $(1/\sqrt3,1)$, so
\eqref{eq:degen2} has exactly one root in $(0,1)$. To thirty digits,
\begin{equation}\label{eq:constants}
\begin{aligned}
  \varrho^*&=0.730471925483789449893373879929\ldots,\\
  \abs{\zeta^*}^2=(\varrho^*)^2&=0.533589233919994846750659957461\ldots,\\
  \lambda^*&=0.978377934627910445119204705630\ldots
\end{aligned}
\end{equation}
Substituting back gives $\varphi(\varrho^*)$ and $\varphi'(\varrho^*)$ of order
$10^{-41}$ at working precision, which is the residual of the root finder.

The turning point $1/\sqrt3$ of $h$ is the same radius that appears in
\cref{prop:saddlenode} as $\varrho_c$ for $n=3$. The coincidence is not accidental, since both are the radius at which
$3(1-\varrho^2)=2$.

\subsection{The saddle-node for a general exponent}
\label{app:D3}

For $\varphi_n(\varrho)=2\lambda \varrho^n+\log(1-\varrho^2)$,
\[
  \varphi_n'(\varrho)=2n\lambda \varrho^{n-1}-\frac{2\varrho}{1-\varrho^2},
  \qquad
  \varphi_n''(\varrho)=2n(n-1)\lambda \varrho^{n-2}-\frac{2(1+\varrho^2)}{(1-\varrho^2)^2} .
\]
Setting $\varphi_n'=0$ gives $\lambda=[nr^{n-2}(1-\varrho^2)]^{-1}$ for $\varrho>0$.
Substituting into $\varphi_n''=0$,
\[
  \frac{2(n-1)}{1-\varrho^2}=\frac{2(1+\varrho^2)}{(1-\varrho^2)^2}
  \quad\Longleftrightarrow\quad
  (n-1)(1-\varrho^2)=1+\varrho^2
  \quad\Longleftrightarrow\quad
  \varrho_c^2=\frac{n-2}{n},
\]
which is \eqref{eq:rc}. For $n=1$ this requires $\varrho_c^2=-1$ and has no
solution, so $\varphi_1$ has no outer critical point for any $\lambda$. For
$n=2$ it gives $\varrho_c=0$, so the outer critical point coincides with the origin
and never separates from it. For $n=3$,
\[
  \varrho_c=\frac{1}{\sqrt3},\qquad
  \lambda_c=\frac{1}{3\varrho_c(1-\varrho_c^2)}=\frac{\sqrt3}{2},\qquad
  \varphi_3(\varrho_c)=\frac13+\log\frac23=-0.0721\ldots<0 .
\]
The outer critical point is therefore born strictly below the value at the
origin, and only at $\lambda^*>\lambda_c$ does it reach it. The interval
$\lambda_c<\lambda<\lambda^*$ is where the second well exists but does not yet
dominate, and where its influence on the asymptotics is exponentially small.

\subsection{The outer hump of the coefficient sum}
\label{app:D4}

The second maximum of the summand $u_j$ noted after \eqref{eq:crossover} is
the outer well of $\Phi_\lambda$ seen in the coefficient variable, and the
correspondence can be made quantitative through \cref{lem:beta}.

A term of index $j$ in \eqref{eq:Ssum2} carries the level $m=3j$, and by
\eqref{eq:betamoments} the level $m$ places its mass at
$\mathbb{E}\abs{\zeta}^2=(m+1)/(m+2k)$. Setting that mean equal to the outer well value
$x=(\varrho^*)^2=0.5335892$ and solving $m+1=x(m+2k)$ for $m$ gives
$m=(2kx-1)/(1-x)$, so to leading order
\begin{equation}\label{eq:humploc}
  m=\frac{(\varrho^*)^2}{1-(\varrho^*)^2}\cdot2k=1.1440328\cdot 2k,
  \qquad
  j=\frac{m}{3}=0.3813443\cdot 2k .
\end{equation}
Direct maximisation of $\log u_j$ at $\epsilon=\epsilon_k$, the root of
$S_k(\epsilon)=1$, locates the second hump at $j/(2k)=0.3812$ for $k=10^4$,
against the predicted $0.3813443$. This is independent of \eqref{eq:humploc},
which is why the small discrepancy is informative; the grid values in
\cref{tab:turns} are computed at $\epsilon=\lambda^*/\sqrt{2k}$ instead and
agree with $x^*$ by construction. The two
determinations use different information, one the phase-space potential and
one the coefficient series, and their agreement is a numerical
confirmation of the identification proved in \cref{app:D5}.

\subsection{Proof of the phase theorem}
\label{app:D5}

We prove \cref{thm:phases}. Write $K=2k$, $\epsilon=\lambda/\sqrt{K}$ and, as
in \eqref{eq:Ssum2}, $S_K=\sum_{j\ge1}u_j$ with
$u_j=\epsilon^{2j}(3j)!\,w_k(3j)/(j!)^2$. Stirling's formula in the form
$\log\Gamma(y)=(y-\tfrac12)\log y-y+\tfrac12\log2\pi+O(1/y)$ applies to each of
the four Gamma factors in $u_j$, whose arguments $3j+1$, $j+1$, $K$ and
$K+3j$ all exceed $\delta K$ when $j\ge\delta K$ and $\delta\le1$. The error is therefore $O(\log K)$ with
a constant depending only on $\delta$, and uniformly for $j=xK$ with
$x\ge\delta$,
\begin{equation}\label{eq:Psi}
  \log u_{xK}=K\,\Psi_\lambda(x)+O(\log K),
\end{equation}
where
\begin{equation}\label{eq:Psidef}
  \Psi_\lambda(x)=2x\log\lambda+3x\log3+x\log x+2x-(1+3x)\log(1+3x).
\end{equation}
The three contributions to \eqref{eq:Psi} are
\begin{align*}
  2j\log\epsilon&=2xK\log\lambda-xK\log K,\\
  \log\frac{(3j)!}{(j!)^2}&=xK\bigl(3\log3+\log x+\log K-1\bigr)+O(\log K),
\end{align*}
and $\log w_k(3j)=-K(1+3x)\log(1+3x)+3xK+O(\log K)$, in which the terms
carrying $\log K$ cancel.

\paragraph{The critical coupling.}
Since $\Psi_\lambda'(x)=2\log\lambda+3\log3+\log x-3\log(1+3x)$ tends to
$-\infty$ at both ends of $(0,\infty)$, the equation $\Psi_\lambda'=0$, that
is
\begin{equation}\label{eq:stationary}
  27\lambda^2x=(1+3x)^3 ,
\end{equation}
has either no positive root or two, the larger of which, $x_+(\lambda)$, is a
local maximum. The two merge when \eqref{eq:stationary} has a double root,
which forces $27\lambda^2=9(1+3x)^2$, that is $1+3x=\sqrt3\lambda$, and
substituting back gives $\lambda=\sqrt3/2$. This is the constant $\lambda_c$
of \cref{prop:saddlenode}.

For $\lambda>\lambda_c$ set $\Theta(\lambda):=\Psi_\lambda(x_+(\lambda))$. By
the envelope theorem $\Theta'(\lambda)=2x_+(\lambda)/\lambda>0$, so $\Theta$
is strictly increasing, and it is negative just above $\lambda_c$ because
$\Psi_{\lambda_c}$ has an inflection at $x_c$ with
$\Psi_{\lambda_c}(x_c)<0$, while $\Theta(\lambda)\to\infty$ as
$\lambda\to\infty$. Hence $\Theta$ has a unique zero, which we call
$\lambda^*$. Imposing $\Psi_\lambda(x)=\Psi_\lambda'(x)=0$ and eliminating
$\lambda$ through \eqref{eq:stationary} leaves
\begin{equation}\label{eq:coeffcrit}
  \log(1+3x)=2x,
  \qquad
  \lambda=\frac{\rme^{3x}}{\sqrt{27x}} ,
\end{equation}
whose unique positive solution is $x^*=0.381344280\ldots$ and
$\lambda^*=0.978377934\ldots$

\paragraph{Identification with the phase-space constants.}
Substituting \eqref{eq:xtor}, so that $1+3x=(1-\varrho^2)^{-1}$ and
$2x=\tfrac{2}{3}\varrho^2/(1-\varrho^2)$, the first equation of \eqref{eq:coeffcrit}
becomes $-\log(1-\varrho^2)=\tfrac{2\varrho^2}{3(1-\varrho^2)}$, which is the first equation of
\eqref{eq:degen}. For the second, $\rme^{2x}=1+3x=(1-\varrho^2)^{-1}$ gives
$\rme^{3x}=(1-\varrho^2)^{-3/2}$, while $27x=9\varrho^2/(1-\varrho^2)$ gives
$\sqrt{27x}=3\varrho(1-\varrho^2)^{-1/2}$, so
$\lambda=\rme^{3x}/\sqrt{27x}=[3\varrho(1-\varrho^2)]^{-1}$, the second equation of
\eqref{eq:degen}. The two critical couplings therefore coincide, and so do
$x^*$ and $\varrho^*$ under \eqref{eq:xtor}. The same substitution turns
$1+3x=\sqrt3\lambda$ into the saddle-node condition of
\cref{prop:saddlenode}.

\paragraph{The two phases.}
Fix $\delta>0$ and split $S_K$ at $j=\delta K$.

For the head, the ratio computed in the proof of \cref{lem:Sk} obeys
\begin{align*}
  \frac{u_{j+1}}{u_j}
  &=\frac{\lambda^2}{K}\cdot\frac{(3j+1)(3j+2)(3j+3)}{(j+1)^2}
   \cdot\frac{K^3}{(3j+K)(3j+K+1)(3j+K+2)}\\
  &\le\ \frac{27\lambda^2(j+1)}{K} ,
\end{align*}
since the first fraction is at most $27(j+1)$ and the second at most one.
For $j\le\delta K$ and $K\ge2/\delta$ the right side is at most
$27\lambda^2\delta+27\lambda^2\delta/2$, so choosing
$\delta\le(81\lambda^2)^{-1}$ makes it at most $1/2$. The head is then a
geometric series and $\sum_{j\le\delta K}u_j\le2u_1=O(1/K)$, using
$u_1=6\epsilon^2\wk(3)$.

For the tail, bounding $3j+K\ge3j$ in the same expression gives
$u_{j+1}/u_j\le\lambda^2K^2(j+1)/j^3\le2\lambda^2K^2/j^2$ for $j\ge1$, so the
terms decay geometrically with ratio at most $1/2$ beyond $J_0:=2\lambda K$,
and
\[
  \max_{j>\delta K}u_j\ \le\ \sum_{j>\delta K}u_j
  \ \le\ \bigl(J_0+2\bigr)\max_{j>\delta K}u_j .
\]
The factor is $O(K)$, so
$\log\sum_{j>\delta K}u_j=\max_{j>\delta K}\log u_j+O(\log K)$, and by
\eqref{eq:Psi} the right-hand side is
$K\max_{x\ge\delta}\Psi_\lambda(x)+O(\log K)$.

If $\lambda<\lambda^*$ then $\Theta(\lambda)<0$, and choosing $\delta$ below
$x_+(\lambda)$ makes $\max_{x\ge\delta}\Psi_\lambda=\Theta(\lambda)<0$, so the
tail is $\rme^{-cK}$ and $S_K=O(1/K)\to0$. If $\lambda<\lambda_c$ there is no
interior maximum at all and $\Psi_\lambda<0$ on $[\delta,\infty)$, with the
same conclusion. If $\lambda>\lambda^*$ then $\Theta(\lambda)>0$ and the tail
alone is at least $\rme^{cK}$, so $S_K\to\infty$.

\paragraph{The constant of \cref{prop:partial}.}
There $\epsilon_k$ is defined by $S_k(\epsilon_k)=1$. Writing
$\lambda_k=\epsilon_k\sqrt{K}$, the two phases just established force
$\lambda_k\to\lambda^*$: if a subsequence had $\lambda_k\le\lambda^*-\eta$
then $S_k\to0$ along it by monotonicity of $S_k$ in $\epsilon$, and if a
subsequence had $\lambda_k\ge\lambda^*+\eta$ then $S_k\to\infty$. Either
contradicts $S_k=1$. \qed

\medskip

At $\lambda=\lambda^*$ the argument gives only $\log S_K=O(\log K)$, since
$\Theta(\lambda^*)=0$. The width of the saddle contributes a factor
$\sqrt{K}$ and the Stirling prefactor a factor $K^{-1/2}$, and these cancel,
so $S_K$ tends to a finite limit. Numerically that limit is $3.272$, giving
the retained fraction $0.234$ quoted after \cref{thm:phases}; we have not
computed it in closed form.

\section{Numerical protocols and reproducibility}
\label{app:numerics}

The numerical statements in this paper are of two kinds. Some check an
identity that is proved elsewhere, and are included because a reader may want
to confirm a constant. Others record the value at the critical coupling in
\cref{thm:phases}, which we do not derive.
This appendix records how each was computed, and one trap that anyone
reproducing the calculations will otherwise fall into.

\subsection{The truncation trap}
\label{app:E1}

The coefficient series \eqref{eq:Ssum2} is not monotone in $j$. Its increment
changes sign twice. The terms decrease from the start, reach a minimum near
$j\approx0.067\,(2k)$, and then increase again to a second maximum near
$j\approx0.381\,(2k)$, which is the outer well located in \cref{app:D4}.
\Cref{tab:turns} gives the two extrema of $u_j$ on the integer grid.

\begin{table}[h]
\centering
\begin{tabular}{rrlrl}
\toprule
& \multicolumn{2}{c}{minimum} & \multicolumn{2}{c}{maximum}\\
\cmidrule(lr){2-3}\cmidrule(lr){4-5}
$k$ & $j$ & $j/(2k)$ & $j$ & $j/(2k)$\\
\midrule
$10^2$ & $14$   & $0.0700$ & $76$   & $0.3800$\\
$10^3$ & $135$  & $0.0675$ & $762$  & $0.3810$\\
$10^4$ & $1342$ & $0.0671$ & $7626$ & $0.3813$\\
\bottomrule
\end{tabular}
\caption{The local minimum and the local maximum of $u_j$ at
$\epsilon=\lambda^*/\sqrt{2k}$, located on the integer grid, that is, the
smallest $j$ at which the increment $\log(u_{j+1}/u_j)$ has changed sign. The
maximum sits at $\lfloor x^*(2k)\rfloor$ with $x^*$ from \eqref{eq:coeffcrit},
which is what the last column records; this is a consistency check on the grid
and not an independent determination of $x^*$, for which see \cref{app:D4}. At
$k=10^2$ the resolution in $j/(2k)$ is $0.005$.}
\label{tab:turns}
\end{table}

The consequence is that a summation routine which stops when the terms become
small deletes the whole outer hump, and therefore deletes the second well.
Solving $S_k(\epsilon)=1$ with such a routine returns a spurious root. The
effect is invisible at small $k$ and appears abruptly, as \cref{tab:trunc}
shows.

\begin{table}[h]
\centering
\begin{tabular}{rll}
\toprule
$k$ & $\epsilon_k\sqrt{2k}$, summing past the hump
    & $\epsilon_k\sqrt{2k}$, early truncation\\
\midrule
$10$    & $0.875666$ & $0.8757$\\
$10^2$  & $0.970338$ & $0.9703$\\
$10^3$  & $0.977614$ & $0.9776$\\
$10^4$  & $0.978302$ & $1.0379$\\
\bottomrule
\end{tabular}
\caption{The root of $S_k(\epsilon)=1$ computed with and without the outer
hump inside the summation window. The two agree until the hump leaves the
window.}
\label{tab:trunc}
\end{table}
The first three rows agree because the hump still lies inside the truncation
window. At $k=10^4$ it does not, and the truncated computation reports a
constant that does not converge to $\lambda^*$. Sums of the form \eqref{eq:Ssum2} must therefore be carried at least to
$j\approx0.4\,(2k)$, and we used $j\le12\,(2k)$ for those. The bound is
specific to the critical scaling $\epsilon\asymp\lambda/\sqrt{2k}$. At fixed
$\epsilon$ the maximal term sits instead at $j\approx\epsilon(2k)^{3/2}$, and
the computations of \cref{app:E3} are carried out over a window around that
point rather than over a fixed multiple of $2k$.

\subsection{Protocols and rates of convergence}
\label{app:E2}

Computations were performed with \texttt{mpmath} at $25$ to $40$ decimal
digits where exact constants were wanted, and in double precision with a
log-sum-exp accumulator where only the location of a maximum or the sign of a
divergence was at issue. Three points recur.

First, the summands of \eqref{eq:Ssum2} and of the norms in \cref{app:A2} span
hundreds of orders of magnitude, so they are accumulated as
$\log\sum\rme^{\log u_j}$ rather than directly. In double precision the naive
sum overflows before the maximal term is reached.

Second, roots of $S_k(\epsilon)=1$ were found by bisection rather than by a
general-purpose root finder. \Cref{lem:Sk} guarantees that $S_k$ is strictly
increasing, which makes bisection unconditionally reliable, whereas a
derivative-based method fails on this function because the derivative is
astronomically large near the root at large $k$.

Third, zeros inside a disk were counted by the argument principle, evaluating
the winding number of the function along the boundary circle with the phase
unwrapped. The circle was taken at radius $0.999$ rather than $1$ to avoid the
boundary, and the coefficient vectors were rescaled to unit maximum before
evaluation, since otherwise the polynomial values overflow.

Two of the tables converge slowly for reasons that are understood, and we
record them so that the residuals are not mistaken for errors.

The pointwise bound of \cref{lem:normal} approaches its limit like
$1/k$: at $\abs{z}=2$ the values of $(1-\abs{z}^2/2k)^{-k}$ at
$k=3,10,100,1000$ are $27.0$, $9.313$, $7.540$ and $7.404$, against
$\rme^{2}=7.389$. This is the correction term $t^2/8k$ identified in the proof.

The constant $\lambda^*$ is approached slowly because \eqref{eq:Psi} carries
an error $O(\log K)$ against a leading term of size $K=2k$, a relative error
$O(\log K/K)$ in the exponent. Correspondingly $\epsilon_k\sqrt{2k}$
in \cref{app:E1} still differs from $\lambda^*=0.978378$ in the fourth digit
at $k=10^4$. \Cref{app:D1} explains why: the bracket in \eqref{eq:crossover}
becomes positive only when $k\gtrsim\rme^{2}\epsilon^{-4}$, which is
$1.2\times10^{6}$ at $\epsilon=0.05$.

\subsection{The surviving exponents}
\label{app:E3}

\Cref{sec:escape} asserts that the modes $\rme^{a_nz^n}$ survive the
contraction exactly for $n\le2$. Setting $a_n=c\,(2k)^{n/2}$ with $c=0.3$ and
computing $\norm{\cdot}_k^2$ from \eqref{eq:key} gives \cref{tab:survive}.
Here $c$ does not scale with $k$. For $n\ge3$ the maximal term then sits at
$j^*\approx c\,(2k)^{n/2}$, far beyond the window used elsewhere in this
appendix, while for $n\le2$ it sits at $j=0$. The entries were obtained by
locating $j^*$ first and summing a window of half-width $8\times10^{5}$ around
it, which is ample: the peak of $\log u_j$ has curvature
$\abs{\rmd^2\log u_j/\rmd j^2}$ at $j^*$ giving a Gaussian width of about
$3\times10^{3}$ in the largest case $n=4$, $k=5000$.

\begin{table}[t]
\centering
\begin{tabular}{rllll}
\toprule
$n$ & $k=5$ & $k=50$ & $k=500$ & $k=5000$\\
\midrule
$1$ & $1.0938$ & $1.0941$ & $1.0942$ & $1.0942$\\
$2$ & $1.1965$ & $1.2414$ & $1.2491$ & $1.2499$\\
$3$ & $3.4004$ & $10^{123.6}$ & $10^{6356}$ & $10^{236705}$\\
$4$ & $10^{11.84}$ & $10^{2356}$ & $10^{257065}$ & $10^{2.6013\times10^{7}}$\\
\midrule
$j^*$ ($n=4$) & $25$ & $2950$ & $299500$ & $29995000$\\
\bottomrule
\end{tabular}

\caption{The squared norm $\norm{\rme^{a_nz^n}}_k^2$ with $a_n=c(2k)^{n/2}$
and $c=0.3$, entries above $10^2$ being given as powers of ten, with the
exponent to four significant digits below $10^3$, as an integer between
$10^3$ and $10^6$, and to five significant digits above. The rows $n=1$ and $n=2$ stabilise, the
latter at $(1-4c^2)^{-1/2}=1.25$, which is the Fock norm of $\rme^{cz^2}$. The
rows $n\ge3$ diverge, faster the larger $n$ is, and the gap between $n=2$ and
$n=3$ is the scaling window of \cref{thm:collapse}. The last line gives the
index of the maximal term for $n=4$, which is what fixes the summation range.
The table is generated by the script of \cref{app:E4} rather than transcribed.}
\label{tab:survive}
\end{table}

\subsection{Availability}
\label{app:E4}

The scripts that reproduce every constant, table and figure coordinate in this
paper, together with the symbolic verifications of \cref{lem:Cs},
\cref{cor:detCs} and \cref{lem:Dr}, are archived at
\texttt{doi:10.5281/zenodo.22313935}.
Each asserts its output against the value printed here, so that
\texttt{verify\_all.py} fails if any number has drifted; it runs $82$ such
checks. \Cref{tab:survive} is generated by one of the scripts and written into
the source rather than transcribed.

\end{document}